\documentclass{LMCS}

\def\dOi{11(1:9)2015}
\lmcsheading%
{\dOi}
{1--26}
{}
{}
{Jan.~\phantom09, 2014}
{Mar.~25, 2015}
{}

\ACMCCS{[{\bf Theory of computation}]: Models of
  computation---Abstract machines; Formal languages and automata
  theory---Grammars and context-free languages} 
{}

\usepackage[latin2]{inputenc}
\usepackage{amssymb}
\usepackage{comment}
\usepackage{graphicx}
\usepackage{hyperref}

\author[E.~Kopczyński]{Eryk Kopczyński}    
\address{Institute of Informatics, University of Warsaw}  
\email{erykk@mimuw.edu.pl}  

\title[Complexity of Problems of Commutative Grammars]{Complexity of Problems of Commutative Grammars\rsuper*}

\def\bbN{{\mathbb N}} 
\def\bbZ{{\mathbb Z}} 
\def\bbQ{{\mathbb Q}} 
\def\bbP{{\mathbb P}} 
\def\bbR{{\mathbb R}} 

\def\calS{{\mathcal S}}
\def\calR{{\mathcal R}}
\def\calZ{{\mathcal Z}}
\def\calF{{\mathcal F}}

\def\sgn{{\rm sgn}}
\def\ra{\rightarrow}
\def\Parikh{\Psi}
\def\supp{{\rm supp}}
\def\set#1{\left\{#1\right\}}
\def\rint#1#2{\left[{#1},{#2}\right]}
\def\rmint#1{\rint{-#1}{#1}}
\def\zint#1#2{\left[{#1}..{#2}\right]}
\def\zmint#1{\zint{-#1}{#1}}
\def\nsum#1{{\left|{#1}\right|}}

\def\nmax#1{{\left|\left|{#1}\right|\right|}}
\def\osum#1#2{#2^{\oplus #1}}
\def\ind{\hskip 3em}

\def\PiP{$\mathrm{\Pi_2^P}$}

\def\Alphabet{\Sigma}     
\def\alphs{A}             
\def\States{Q}            
\def\qstate{N}            
\def\Trans{\delta}        
\def\trans{d}             
\def\transf{\partial}     
\def\qini{q_I}            
\def\mv#1#2#3{#1\stackrel{#2}{\rightarrow}#3}
\def\mvm#1#2#3{#1\stackrel{#2}{\rightarrow}#3}
\def\source{{\rm source}}
\def\target{{\rm target}}
\def\out{\Parikh}
\def\under{{\rm U}}

\def\Regions{\calR}  
\def\Reg{{\rm Reg}}       
\def\reg{{\rm reg}}       

\def\hadamard#1#2{H_{#1}\left({#2}\right)} 
\def\gam#1{\Gamma_G^{#1}}
\def\frunbound{B_G}
\def\frunboundx{C}
\def\grammarboundl{B_{\ref{normalcommlemma}}}
\def\grammarbound{B_{\ref{normalcommcorol}}}
\def\secv{B_{\ref{secbound}}}
\def\taulimit{{C_{\ref{regionshape}}}}
\def\reduceto{{B_{\ref{reducerlemma}}}}
\def\regbaselimit{{B_{\ref{regionshape}}}}
\def\phibound{C_\phi}
\def\aphibound{\alphs\phibound}
\def\mainboundlimit{{C_{\ref{mainbound}}}}

\def\RunsTable{\mathcal R}
\def\PathTable{\mathcal P}

\def\Cout{\calZ}          

\newcounter{mycount}[section]

\newtheorem{lemma}[thm]{Lemma}
\newtheorem{theorem}[thm]{Theorem}
\newtheorem{corol}[thm]{Corollary}
\newtheorem{proposition}[thm]{Proposition}

\def\myqed{}

\begin{document}

\begin{abstract}
\noindent We consider commutative regular and context-free grammars, or, in other words,
Parikh images of regular and context-free languages. By using linear algebra and 
a branching analog of the classic Euler theorem, we show that, under an assumption
that the terminal alphabet is fixed, the membership problem for regular grammars
(given $v$ in binary and a regular commutative grammar $G$, does $G$ generate $v$?)
is P, and that the equivalence problem for context free grammars (do $G_1$ and $G_2$
generate the same language?) is in \PiP.
\end{abstract}

\titlecomment{{\lsuper*}This paper is the full version of the author's part of \cite{licsfull}, with some differences.}

\maketitle

\section{Introduction}

Let $\Alphabet$ be a finite alphabet. By $\Alphabet^*$ we denote the set of
words over $\Alphabet$, or finite sequences of elements of $\Alphabet$. For
a word $w \in \Alphabet^*$, by $\Parikh(w)$ (the Parikh image of $w$) we
denote the function from $\Alphabet$ to non-negative integers $\bbN$, such that
each $x \in \Alphabet$ appears $\Parikh(w)(x)$ times in $w$. For a language
$L \subseteq \Alphabet^*$, $\Parikh(L) = \set{\Parikh(w): w \in L} \subseteq \bbN^\Alphabet$.

Context free and regular languages are one of the most important classes of languages in
computer science \cite{hopcroft79}. By a famous result of Parikh \cite{parikh}, a subset 
of $\bbN^\Alphabet$ is a Parikh image of a context free language if and only if
it is a \emph{semilinear set}, or a union of finitely many \emph{linear sets}.

In this paper, we explore the complexity of various problems related to
Parikh images of context free languages, such as the following:

\begin{itemize}
\item Membership: Given a context-free grammar $G$ and $v \in \bbN^\Sigma$ (given in binary). Is
$v$ a member of the $\Psi(G)$, the Parikh image of the language generated by $G$?
\item Universality: Given two context-free grammars $G$, is $\Psi(G)$ equal
to $\bbN^\Sigma$?
\item Inclusion: Given two context-free grammars $G_1$ and $G_2$, does $\Psi(G_1) \subseteq \Psi(G_2)$?
\item Equality: Given two context-free grammars $G_1$ and $G_2$, does $\Psi(G_1) = \Psi(G_2)$?
\item Disjointness: Given two context-free grammars $G_1$ and $G_2$, is $\Psi(G_1) \cap \Psi(G_2)$ nonempty?
\end{itemize}

Since in this paper we are never interested in the order of terminals or non-terminals, we 
treat everything in a commutative way. This allows us to identify the commutative
languages (subsets of $\Alphabet^*$) with their Parikh images (subsets of $\bbN^\Alphabet$).

In the non-commutative case, the size of alphabet usually does not matter very much: 
larger alphabets can be encoded as words over smaller alphabets, for example the
alphabet $\set{a,b,c}$ can be encoded as $\set{b, ba, baa}$ or $\set{a, ba, bb}$.
This changes in the
commutative case: each new letter in the alphabet literally brings a new dimension to
the Parikh image. If we do not fix the size of the alphabet, it can be easily shown
that even the membership problem for regular languages is NP complete.

There are many practical uses of regular and context-free languages which
do not care about the order of the letters in the word. For example, 
when considering regular languages of trees, we might be not 
interested in the ordering of children of a given node. \cite{xml} and
\cite{xml2} consider XML schemas allowing marking some nodes as unordered.

Some complexity results regarding semilinear sets and
commutative grammars have been obtained by
D. Huynh \cite{semipi,semipi2}, who has shown that equivalence is \PiP-hard
both for semilinear sets and commutative grammars  (where \PiP\ is the dual
of the second level of the polynomial-time hierarchy, \cite{polyh}).

Some research has also been done in the field of communication-free
Petri nets, or Basic Parallel Processes (BPP). We say that a Petri net 
(\cite{petri1}, \cite{petri2}) is communication-free if each transition has
only one input. This restriction means that such a Petri net is essentially
equivalent to a commutative context-free grammar. \cite{hcyen} shows that
the reachability equivalence problem for BPP-nets can be solved in
$DTIME\left(2^{2^{ds^3}}\right)$, where $d$ is a constant and $s$ is the
size of the problem instance. For general Petri nets, reachability
(membership in terms of grammars) is decidable
\cite{kosaraju, lerouxpetri}, although the known algorithms
require non-primitive recursive space; and reachability equivalence
is undecidable \cite{hack26}. Also, some harder types of equivalence problems
are undecidable for BPP nets \cite{huttel}. See \cite{esparza2} for a survey
of decidability results regarding Petri nets.

As mentioned above, we will assume in this paper that the alphabet is fixed. 
We have the following two main results:

\begin{itemize}
\item The
membership problem for commutative regular languages over a fixed alphabet is in P.
This was open
for a long time (even for a binary alphabet), until it was solved independently
by the author of this paper and
Anthony Widjaja Lin. This result,
and its applications, was presented as a merged paper at the LICS conference
\cite{licsfull}. 

\item The equivalence problem for commutative context-free languages over a fixed
alphabet is in \PiP. As far as we know, there
have been no successful previous attempts in this direction, except for the
much simpler case where the alphabet has only one symbol \cite{hyunh1}.
\end{itemize}

\noindent In context free grammars, usually the order of transitions used in a derivation is
important: we are never allowed to use a non-terminal which has not yet been produced.
For example, a transition $X \ra aX$ allows us to produce an arbitrary amount of
the terminal symbol $a$, but only if we have access to the non-terminal symbol $X$.
However, derivations also can be defined commutatively. 
The well known Euler's theorem gives a necessary and sufficient condition for whether
there is a path or cycle in a graph which uses each edge $e$ exactly $n_e$ times: the
condition says that each vertex has to be entered and left exactly the same number of
times (\emph{Euler condition}) and reachable from the starting vertex
(\emph{connectedness}). Thus, we can forget the order of edges on such a path or
cycle, count them, and just check that the conditions are satisfied; a roughly similar
approach is
used in the definition of a cycle in the construction of homology groups in algebraic
topology
\cite{hatcher}. The same approach works in our more general case --- we can define a
\emph{commutative run} and a \emph{commutative cycle} by counting the number of times
each transition has been used, and just like in Euler's theorem, 
a very simple condition can be used to check whether
such a function from transitions to integers is indeed a Parikh image of a valid
derivation. Similar technique has been used 
previously in \cite{espfund}; it also has been used successfully to solve
an open problem in database theory \cite{tits}.


Although we are not allowed to use a non-terminal which has not yet been produced,
our framework straghtforwardly allows the following interesting generalization:
we allow our terminal symbols to be produced in negative quantities. In this case,
these negative productions do not have to be balanced by positive productions.
Thus, our commutative languages are interpreted as subsets of $\bbZ^\Sigma$ rather than
$\bbN^\Sigma$. Although such languages do not commonly appear in the theory of languages,
they turn out to be interesting and useful. For example, such negative production allow
to reduce disjointness of $A$ and $B$ to membership very easily -- just check for
membership of 0 in $A-B = \{a-b: a\in A, b\in B\}$
(see Proposition \ref{disunfix} below)
.
This generalization was also 
discovered independently by Anthony Widjaja Lin.

\subsection{Related papers}
This paper is the full version of \cite{licsfull}, with the following major differences:
\begin{itemize} 
\item \cite{licsfull} is a result of merging of two
submissions. This paper includes only results obtained by the author.

\item Equivalence of commutative context-free grammars (Theorem \ref{normalcommcorol})
has been generalized to the case where we allow terminal symbols to be produced in
negative quantities.
\item Universality of commutative context-free languages has been only proven to
be in \PiP\ in \cite{licsfull}. Here, we prove that it is in fact \PiP-complete.
\item Full proofs are included.
\end{itemize}

\subsection{Structure of the paper}
The paper is structured as follows.

Section \ref{secprelim} introduces basic definitions and facts from linear algebra and
language theory, and then presents bounds on the size of runs and cycles of commutative
grammars. The techniques used to obtain these bounds are almost the same for regular
and non-regular grammars, but the results are much stronger in the regular case.
This section culminates in Theorem \ref{decomposition}, 
which gives a compact representation of a commutative regular
language (equivalently, a Parikh image of a regular language).
This compact representation will be used in Section \ref{seccomplexity}
to solve the membership problem of commutative regular languages in P.

The whole Section \ref{seccommg} presents a ``window theorem'' (Theorem
\ref{normalcommcorol}) for (non-regular) commutative grammars.
This theorem roughly says that, in order to decide whether two
grammars are equivalent, we only have to look at a window of exponential size,
and it will be used in Section \ref{seccomplexity} to show that
the equivalence of commutative grammars is in \PiP. In fact, 
most of Section \ref{seccommg} is the proof of Lemma \ref{normalcommlemma}
about semilinear sets, which does not refer commutative grammars directly;
Theorem \ref{normalcommcorol} is a simple corollary.
Lemma \ref{normalcommlemma}, as well as some other lemmas in this section,
thus could potentially have applications to semilinear sets in general, whether
they come from commutative grammars or not.

Section \ref{hardgrammar} presents a non-regular grammar
over a three letter alphabet which has no compact representation similar to one 
given by Theorem \ref{decomposition}. We believe this example is of interest,
because it shows why we cannot prove Theorem \ref{normalcommcorol}
using a simpler approach.

Section \ref{seccomplexity} gives the answers to 
questions about the complexity of the problems mentioned in the introduction; 
some results are obtained directly, and for some we need our compact representation and
window theorems from Sections \ref{secprelim} and \ref{seccommg}.

There is a conclusion in Section \ref{secconclusion}.

A reader interested only in the first main result (that 
the membership problem of commutative regular languages in P) only has to read
Section 2 (assuming that the grammar is regular)
and the respective part of Section 5. 
Due to the lack of branching,
results in subsections \ref{subs_subrun_order} and \ref{subs_derivation} can
be obtained in a much more straightforward way for regular grammars.
A reader interested in the second
main result (that the equivalence problem for commutative context-free languages over
a fixed alphabet is in \PiP) has to read Section 2 (without concentrating on regular
grammars), Section 3, and the respective part
of Section 5; Section 4 should also be of interest for such readers.

\section{Preliminaries}\label{secprelim}

\subsection{Vectors and matrices, and semilinear sets}

In this subsection we present the basic notation, notions and facts from linear algebra,
which will be used throughout the paper.

As usual, we denote the set of integers by $\bbZ$, the set of non-negative integers
by $\bbN$, the set of rational numbers with $\bbQ$, and the set of real numbers with
$\bbR$. We also denote the set of
non-negative real numbers with $\bbP$ (we do not use the more standard notation
$\bbR^+$ to avoid double indexing).

We also use the notation $\rint{a}{b}$ for the interval of all real numbers between
$a$ and $b$, and $\zint{a}{b}$ for the interval of all integers between $a$ and $b$.

Let $X$ be a finite set. We interpret $\bbN^X$ as multisets over $X$: $z \in \bbN^X$
represents a multiset where each $x \in X$ appears exactly $z_x$ times. 
For a $x \in X$, by $[x] \in \bbN^X$ we represent the multiset which
represents $x$: $[x](y) = 1$ iff $y=x$, 0 otherwise. The set $\bbN^X$ is 
a subset of $\bbZ^X$ (intuitively, we allow the elements of $X$ to appear in
negative quantities), $\bbQ^X$, $\bbP^X$, and $\bbR^X$.

For $v \in \bbR^X$, let $\nsum{v} = \sum_{x \in X} |v_x|$, and $\nmax{v} = \max_{x \in X} |v_x|$.

We also interpret elements of $\bbR^X$ as (column) vectors. 
By $\bbR^{X\times Y}$ we denote the set of matrices over $\bbR$ with
columns indexed by $X$, and rows indexed by $Y$; $A^{X \times Y} \subseteq
\bbR^{X \times Y}$ denotes matrices where all entries are in $A \subseteq \bbR$.
For $M \in \bbR^{X \times Y}$, we use
the notation $M^x_y$ for the coefficient in row $y \in Y$ and column $x \in X$. For
$M \in \bbR^{X \times Y}$, $N \in \bbR^{Y \times Z},$ and $v \in \bbR^X$ we define
the multiplication in the usual way:

\[ Mv \in \bbR^Y,\ (Mv)_y   = \sum_{x \in X} M^x_y v_x\]
\[ NM \in \bbR^{X \times Z}, \ (NM)^x_z = \sum_{y \in Y} N^x_y M^y_z\]

For a matrix $M \in \bbR^{X \times Y}$ and a set $S \subseteq \bbR^X$, by $MS \subseteq \bbR^Y$ we denote
$\set{Ms: s \in S}$. Similarly, for sets $S_1, S_2 \subseteq \bbR^X$, we define
$S_1+S_2 = \set{s_1+s_2: s_1 \in S_1, s_2 \in S_2}$ and 
$S_1-S_2 = \set{s_1-s_2: s_1 \in S_1, s_2 \in S_2}$ (the algebraic sum and difference).

For a finite set $S \subseteq \bbR^X$ and $A \subseteq \bbR$, we define

\[\osum{A}{S} = \set{\sum_{s \in S} sa_s : \forall s\ a_s \in A}.\]
Typically, $A$ will be one of $\bbZ$ ($\osum{\bbZ}{S}$ is the additive group
generated by $S$), $\bbN$ ($\osum{\bbN}{S}$ is the additive monoid),
$\bbR$ ($\osum{\bbR}{S}$ the linear space spanned by $S$), 
$\bbP$ ($\osum{\bbP}{S}$ is the cone spanned by $S$),
or $\zint{0}{H}$ (we limit the number of uses of
each element of $S$).

A set $V \subseteq \bbZ^X$ is {\bf linear} iff it can be written as
$V = v_0 + \osum\bbN P$, where $v_0 \in \bbZ^X$ and the set $P \subseteq \bbZ^X$ is finite.
The vector
$v_0$ is called the {\bf base} of $V$, and the elements of $P$ are called {\bf periods}.
Moreover, we say that a linear set $V$ is a {\bf simple linear set} iff the elements of $P$
are linearly independent.

A set $V \subseteq \bbZ^X$ is {\bf semilinear} iff it is a union of finitely many linear sets.

A semilinear set $V \subseteq \bbZ^X$ is called a {\bf simple bundle} iff it can be
written as $V = W + \osum\bbN P$, where $W \subseteq \bbZ^X$ is finite and the elements
of $P$ are linearly independent. We say that a bundle is {\bf bounded} by $(B,Y)$ iff 
$\nmax{v} \leq B$ for each $v \in W$, and $\nmax{p} \leq Y$ for each $p \in P$.

The following lemmas from the linear algebra will be important for us. We denote
the determinant of $M \in \bbR^{X \times X}$ by $\det M$.

\begin{lemma}\label{smalldet}
Let $M \in \rmint{C}^{X \times X}$ be a matrix. Then the determinant of $M$ is bounded
by $\hadamard{|X|}{C} = |X|! C^{|X|}$. Moreover, if $M \in \bbZ^{X \times X}$,
the determinant of $M$ is an integer.
\end{lemma}

\begin{proof}
This follows straightforwardly from the Leibniz formula for determinant. In fact,
Hadamard \cite{hadamard} has shown a better bound: 
$\hadamard{|X|}{C} = C^{|X|} |X|^{|X|/2}$. \myqed
\end{proof}

Let us recall the well known Cramer's rule \cite{cramer}:

\begin{lemma}\label{cramer}
Let $A \in \bbR^X_X$ be a non-degenerate matrix, and $b \in \bbR^X$. Then the system
of equations $Av=b$ has a unique solution given by $v_x = \det(A_x) / \det(A)$, where
$A_x$ is obtained by replacing the $x$-th column of $A$ with the vector $b$.
If $A$ is not a non-degenerate matrix, the system of equations either has no solution,
or infinitely many solutions.
\end{lemma}

The following corollary is straightforward:
\begin{lemma}\label{syseq}
Let $X$ be a fixed set of indices, and $A \in \bbN$. There is a number $B$, polynomial
in $A$, such that whenever the system of equations $Mx=v$, where
$M \in [-A..A]^{X\times X}$, $v \in [-A..A]^X$ has a unique solution $x\in\bbR^X$,
we have $x = x'/C$, where $x' \in [-B..B]^X$, and $C \in [-B..B]$.
\end{lemma}
\begin{proof}
We apply the Cramer's rule (Lemma \ref{cramer}) and the Hadamard bound (Lemma \ref{smalldet}).\myqed
\end{proof}

\begin{lemma}\label{detgroup}
Let $M \in [-C..C]^{X \times X}$ be a non-degenerate matrix. Then \[(\det M) \bbZ^X \subseteq M \bbZ^X.\]
\end{lemma}

\begin{proof}
Again, this is immediate from Cramer's rule (Lemma \ref{cramer}).
The intuition is as follows.
For $u, v \in \bbZ^\Alphabet$, we say that $u \equiv v$ iff $u - v \in M \bbZ^X$.
The quotient group $\bbZ^\Alphabet /_\equiv$ has $\det M$ elements (intuitively,
for $|X| = 2$, the number of elements is equal to the area of the
parallelogram given by columns of $M$; this intuition also works in other
dimensions). Thus, $(\det M) v \equiv 0$.\myqed
\end{proof}

\begin{lemma}\label{detrewrite}
Let $P \subseteq [-C..C]^X$ be a linearly dependent set of vectors. Then for some
$\alpha \in \bbZ^P$ we have $\sum_{v \in P} \alpha_v v = 0$, where $\nmax{\alpha} \leq 
\hadamard{|X|}{C}$, and $\alpha_v > 0$ for some $v \in P$.
\end{lemma}

\begin{proof}
Without loss of generality we can assume that $P$ is a minimal linearly dependent set.
Let $(\beta_v) \in \bbR^P$ be the set of coefficients for which $\sum_{v\in V} \beta_v v=0$,
and such that $\beta_v\neq 0$ for some $v \in P$. Since $P$ is minimal, we know that
$(\beta_v)$ is unique up to a constant: that is, if $(\beta_v)$ and $(\beta'_v)$ have
this property, then $\beta'_v = q \beta_v$ for some $q \in \bbR$.

Let $u \in P$ be the element such that $|\beta_u| \geq |\beta_v|$ for each $v \in P$.
Let $P_0 = P - \set{u}$. We know that the $\osum{\bbR} P_0$ = $\osum{\bbR} P$.

Since $P$ was minimal, we know that $P_0$ is linearly independent set.
Let $P_1 \supseteq P_0$ be such that $P_1 \subseteq [-C..C]^X$, and $P_1$ is a base
of $\bbR^X$. This can be done by extending $P_0$ by unit vectors which are not yet in
the subspace spanned by $P_0$.

Let $M \in [-C..C]^{X \times X}$ be a matrix whose columns are the elements of $P_1$.
Since $P_1$ is a base, each vector $v$ can be written as a linear combination of
elements of $P_1$ with real coefficients in a unique way.
From Lemma \ref{detgroup} we know that for $v=(\det M)u$ the coefficients are
integers; moreover, since $(\det M)u \in \osum{\bbR} P = \osum{\bbR} P_0$,
the coefficients
are 0 for elements of $P_1 \backslash P_0$. Thus, we get $\sum_{v \in P} \alpha_v v = 0$,
where $\alpha_u = \det M$
and $\alpha_v$ are integers. Since $(\beta_v)$ was unique up to a constant, and $u$
was the vector with the largest coefficient, we know that the same holds for
$(\alpha_v)$. From Lemma \ref{smalldet} we know that
$\alpha_v \leq \det M \leq \hadamard{|X|}{C}$ for all $v \in P$. \myqed
\end{proof}

\begin{lemma}\label{detrewritemuch}
Let $P \subseteq [-C..C]^X$ be a finite set of vectors, and $w \in \osum\bbN P$. 
Let $H = \hadamard{|X|}{C}$. Then there is a linearly independent set
$P_0 \subseteq P$ such that $w \in \osum{[0..H]}P + \osum\bbN P_0$. In other words,
we can assume that all the periods except ones from $P_0$ are multiplied by factors
bounded by $H$.
\end{lemma}

\begin{proof}
Assume that $w = \sum_{v \in P} n_v v$. Let $H = \hadamard{|X|}{C}$.

Let $P^* = \set{v \in P: n_v \geq H}$. Induction over the cardinality of $P^*$.

If $P^*$ is linearly independent, we are done (we can simply set $P_0 = P^*$).

Otherwise, from Lemma \ref{detrewrite}
we know that there is a set of coefficients $(\alpha_v) \in \bbZ^{P_0}$ such that
$\sum_{v \in P_0} \alpha_v v = 0$, and $\alpha_v \leq H$ for each
$v \in P_0$, and at least one $\alpha_v$ is positive.
For each $v \in P_0$, we subtract $\alpha_v$ from $n_v$; this does not effect the
equation $w = \sum_{v \in P} n_v v$. We repeat this until we get $n_v - \alpha_v < 0$
for some $v \in P_0$ (this will have to eventually happen, since at least one $\alpha_v$
is positive). Since $n_v - \alpha_v < 0$ and $\alpha_v \leq H$, we get that $n_v < H$.
Thus, the new set $P^*$ is a subset of the old one, and we can apply the induction
hypothesis. \myqed
\end{proof}

\subsection{Commutative grammars, runs, and cycles}

A commutative grammar (see for example \cite{esparza2}) is a tuple $G = (\Alphabet, \States, \qini, \Trans)$, where $\Alphabet$ is
a finite alphabet of {\bf terminal symbols}, $\States$ is a set of {\bf non-terminal symbols},
$\qini \in \States$ is the {\bf initial non-terminal symbol}, and $\Trans \subseteq \States \times \bbZ^\Alphabet \times \bbN^\States$ is the
{\bf transition relation}. We will write the transition $(q,a,t) \in \Trans$ as
$\mv{q}{a}{t}$, or using a multiplicative notation: $\mvm{q}{xy^{-1}}{q_1q_2^2}$ denotes
the transition $(q, [x]-[y], [q_1]+2[q_2])$. The intuition here is that our commutative
grammar allows us to use a
non-terminal symbol $q$ to produce the multiset of terminal and non-terminal symbols,
given by $a$ and $t$; we allow our terminal symbols to be produced in negative quantities.
For a transition $\trans = (q,a,t) \in \Trans$, we denote $q$, $a$ and $t$ by $\source(\trans)$,
$\out(\trans)$ and $\target(\trans)$, respectively.

We say that a grammar is {\bf positive} iff $\out(\trans) \geq 0$ for each $\trans \in \Trans$.
Note that our definition is more general than the usual one (e.g., \cite{esparza2}),
since we allow non-positive commutative grammars.

Runs, paths, and cycles are usually defined as sequences or trees of transitions.
However, we will define them as multisets of transitions. The next subsection will
show how our definitions are related to the usual ones.

Let $R \in \bbN^\Trans$ be a multiset of transitions. Let 
$\source(R) = \sum_{\trans \in \Trans} R_\trans\ [\source(\trans)]$,
$\out(R) = \sum_{\trans\in\Trans} R_\trans\ \out(\trans)$,
$\target(R) = \sum_{\trans \in \Trans} R_\trans\ \target(\trans).$
Intuitively, the number $\source(R)$ counts times each non-terminal symbol has been
used, while $\target(R)$ counts the number of times each non-terminal symbol has been
produced, and $\out(R)$ counts the number of times each terminal symbol has been
produced. For $p, q \in \States$ we say that $p \ra_R q$ iff for some $\trans \in \Trans$ such
that $R_\trans>0$, $p = \source(\trans)$ and $\target(\trans)(q) > 0$; i.e., a transition which produces
$q$ from $p$ appears in $R$ with a non-zero quantity. By $\ra^*_R$ we denote the
transitive reflexive closure of $\ra_R$. For a multiset $s \in \bbN^\States$ and 
$q \in \States$, we say that $s \ra^*_R q$ iff 
$p \ra^*_R q$ for some $p \in s$.
By $\supp(R) \subseteq \States$ we denote the set of non-terminals
$q \in \States$ such that $(\source(R))_q > 0$.

We say that a multiset of transitions $R \in \bbN^\Trans$ is a {\bf (commutative) subrun} from
$s \in \bbN^\States$ to $t \in \bbN^\States$
if the following two conditions are satisfied:

\begin{itemize}
\item {\it Euler condition:} $\source(R) - s = \target(R) - t$;
\item {\it Connectivity:} whenever $q \in \source(R)$, we have $s \ra^*_R q$.
\end{itemize}

Intuitively, in a subrun, for each state $q$, we
start with $s_q$ copies of state $q$, 
during the run we use up $\source_q(R)$ copies and produce
$\target_q(R)$ more copies, and $t_q$ copies remain at the end.
The Euler condition says that these numbers
agree. In the case where $s=t$ it means that each state is produced exactly as
many times as it is produced, just like in the classic Euler condition which says
that each vertex has the same in-degree and out-degree. In the next subsection we
will show that, as long as our conditions on a commutative subrun are satisfied,
there is an ordering such that we never use up something which has not yet been
produced.

The following kinds of subruns will be of most interest for us:
\begin{itemize}
\item A {\bf run from $p$} is a subrun from $[p]$ to 0.
\item A {\bf path from $p_1$ to $p_2$} is a subrun from $[p_1]$ to $[p_2]$.
\item A {\bf cycle from $p$} is a path from $p$ to $p$.
\end{itemize}

By $\out(G)$ we denote the (commutative) language of $G$, or the set of $\out(R)$
for all commutative runs $R$ from $\qini$.

We say that $G$ is {\bf in normal form} iff for each transition $(q,a,t)$ in $\Trans$
we have $|t| \leq 2$ and $|a| \leq 1$. Moreover, we say that $G$ is {\bf regular}
iff for each transition $(q,a,t)$ in $\Trans$ we have $|t| \leq 1$ and $|a| \leq 1$.
If $G$ is not in normal form, it is straightforward to construct a grammar $G'$ in
normal form such that $\out(G) = \out(G')$. This is done by replacing each transition
which does not satisfy the restriction by several simpler transitions, adding additional
non-terminals. For example, $\mvm{q}{a_1a_2}{q_1q_2q_3}$ is replaced
with $\mvm{q}{a_1}{q_1q'}$ and $\mvm{q'}{a_2}{q_2q_3}$, where $q'$ is
a new non-terminal. 

Usually, we will assume that our grammars are in normal form; in this case,
the size of the grammar can be described by stating the number of non-terminals $N$
and the size of the alphabet $A$, 
since for a grammar in normal form the number of transitions is bounded by $O(N^3 \alphs)$.
In the sequel, many results will work both for regular grammars and the general case,
but be much stronger in the regular case -- as a typical example, for regular grammars
we can achieve a polynomial bound, but for non-regular grammars in normal form we can
only achieve an exponential one. Both the
polynomial bound in the regular case and the exponential bound in the general case will
be of interest for us.
The notation $\gam N$ introduced in Lemma
\ref{gammalemma} below, which is linear in $N$ for a regular $G$ and exponential
for a non-regular $G$, should make it clear that a result is given both for regular
and general grammars, but it is stronger for regular ones.

We say that a cycle $C$ is {\bf simple} iff it is not a sum of two smaller non-zero
cycles. We say that a run $R$ is a {\bf skeleton run} iff it cannot be written as
$R = R_0+C$, where $C$ is a cycle, and $\supp(R)$ = $\supp(R_0)$.


\subsection{Commutative subruns can be ordered}
\label{subs_subrun_order}

In this subsection we show that our commutative subruns can be ordered
correctly, that is,
in each non-empty commutative subrun from $s$ we can choose the first transition $\delta$ such
that the remaining part is a subrun from $s - \source(\delta) + \target(\delta)$;
in other words, we can order the transitions in a subrun in such a way that, if 
we start from $s$, and
each transition consumes $\source(\delta)$ and produces $\target(\delta)$, we can
arrange the transitions in a way that we never consume something which has not been
produced yet.
This result is practically
equivalent to Theorem 3.1 in \cite{espfund}, where it has been stated in the
setting of communication-free Petri nets. Also note that, for regular grammars
and $\nsum{s} = 1$, this can be seen as a restatement of the well known Euler's theorem.

\begin{theorem}\label{subrun_lemma}
If $R$ is a non-empty subrun from $s$ to $t$, then there is a $\delta \in R$ such that
$\source(\delta) \in s$, and
$R' := R-\delta$ is a subrun from $s' := s-\source(\delta)+\target(\delta)$ to $t$.
\end{theorem}

\begin{proof}
There are two cases:
\begin{itemize}
\item the relation $\ra_R$ has a cycle starting in $s$, i.e., there is a $p \in s$ such
that $p \ra_R q$ and 
$q \ra^*_R p$. In this case, there is a $\delta\in R$ such that $\source(\delta) = p$
and $\target(\delta) = q$. We have to show that $R'$ is a subrun from $s'$ to $t$.
Let $r \in \source(R')$. Since $R$ was a subrun from $s$,
we know that $s \ra^*_{R} r$, and we have to show that
$s' \ra^*_{R'} r$. Take a minimal path $\pi$ witnessing $s \ra^*_{R} r$. 
Since $\pi$ is minimal, only its first transition starts from
an element of $s$, and in particular, can be equal to $\delta$.

\begin{itemize}
\item If $\pi$ did not start in $p$, it is still a path for $\ra_{R'}$ ($\delta$ was not in $\pi$).
\item If $\pi$ started with $\delta$, we just remove $\delta$
from $\pi$, thus obtaining a path witnessing $s' \ra^*_{R} r$. 
\item If $\pi$ started in $p$ but not with $\delta$, we know that $q \ra^*_R p \ra^*_R r$,
and both subpaths do not contain $\delta$ (provided that we take a minimal path
$q \ra^*_R p$) and thus they are still paths for $\ra_{R'}$.
\end{itemize}

\item otherwise, take any $\delta$ such that $\source(\delta) = p \in s$. Such a
$\delta$ must exist from the connectedness condition. Again, we have
to show that $R'$ is a subrun from $s'$ to $t$. 
Let $r \in \source(R')$. Since $R$ was
a subrun from $s$, we know that $s \ra^*_{R} r$, and we have to show that
$s' \ra^*_{R'} r$. Take a minimal path witnessing $s \ra^*_{R} r$. 
Since $\pi$ is minimal, only its first transition starts from
an element of $s$, and in particular, can be equal to $\delta$.
There are four subcases:

\begin{itemize}
\item If $\pi$ did not start in $p$, it is still a path for $\ra_{R'}$ ($\delta$ was not in $\pi$).
\item If $\pi$ started with $\delta$, we just remove $\delta$ from $\pi$,
thus obtaining a path witnessing $s' \ra^*_{R} r$.
\item If $\pi$ started in $p$ but not with $\delta$, and $s_p \geq 2$, then
$s'(p) \geq 1$, so $p \in s'$ and the path is still valid for $R'$.
\item If $\pi$ started in $p$ with $\delta' \neq \delta$, and $s_p = 1$,
we know that $\source_p(R) - s_p = \target_p(R) - t_p$. Since $\source(\delta) =
\source(\delta') = p$ and $s_p = 1$, $\target_p(R) > 0$, and thus
there must be a transition $\delta''$ such that
$p \in \target(\delta'')$. We know that $s \ra^*_R \source(\delta'')$, and 
the path cannot
include $\delta$ -- otherwise, we get a sequence of paths
$s \ra^*_R p \ra_R \target(\delta) \ra_R^* \source(\delta'') \ra_R p$, and thus
we get a cycle, which was dealt with in the previous case. Thus, we have
$s \ra^*_{R'} \source(\delta'') \ra_{R'} p \ra^*_{R'} r$.\qedhere
\end{itemize}
\end{itemize}
\end{proof}

\subsection{Derivation trees}
\label{subs_derivation}

In this subsection we compare our commutative runs and cycles with the usual derivation
trees. We obtain a subrun from a derivation tree simply by counting how many times
each transition has been used; using the result of the previous section,
we also show that this process can be reversed,
i.e.,~a derivation tree exists for each commutative subrun. 
We then use the derivation
trees to show upper bounds on the size of simple cycles and skeleton runs.

Again, this is much easier for regular grammars --- in this case, derivation trees are
simply paths, i.e., sequences of transitions $\transf_0, \transf_1, \transf_2, \ldots, \transf_d$ such
that $[\source(\transf_{i+1})]$ equals $\target(\transf_i)$.

Let $G$ be a commutative grammar in normal form. A {\bf derivation tree} from $p \in \States$
is a tuple $T = (V, v_0, P, \transf)$, such that:

\begin{itemize}
\item $V$ is an arbitrary set of vertices,
\item $v_0 \in V$ is a special vertex, called the root of $V$,
\item $P$ is a function from $V-\set{v_0}$ to $V$ (parent), such that for each $v \in V$
there is a $n \in \bbN$ (called the {\bf depth} of $v$) such that $P^n(v) = v_0$,
\item $\transf$ is a function from $V$ to $\Trans$. We will use $\source(v)$ and $\target(v)$
for $\source(\transf(v))$ and $\target(\transf(v))$, respectively.
\item $\source(\transf(v_0)) = p$,
\item for each $v \in V$, we have $F(v) \geq 0$, where
\[F(v) = \target(v) - \sum_{w: P(w) = v} [\source(w)].\]
\end{itemize}

We also denote $F(T) = \sum_{v \in V} F(v)$, $\source(T) = \sum_{v \in V} [\source(v)]$,
and $\target(T) = \sum_{v \in V} \target(v)$.

Intuitively, each vertex $v$ of the derivation tree represents that we are using a
non-terminal and produce new terminals and non-terminals, according to the transition
$\transf(v)$.
Children of $v$ can use the non-terminals produced. $F(v)$ represents the ``free''
non-terminals
which have been produced, but have not been used by children; $F(v) \geq 0$ represents
the fact that children cannot use non-terminals which have not been produced.

Let $\under(T): \Trans \ra \bbN$ be the function counting the number of times
each transition has been used in the derivation tree $T$: 
$(\under(T))_\trans = |\set{v \in V: \transf(v) = \trans}|$.

We say that a derivation tree $T$ from $p$ is {\bf full} iff $F(T) = 0$, and a {\bf path to $p_2$}
iff $F(T) = [{p_2}].$ We say that a derivation tree $T$ is {\bf cyclic} from $p$
iff it is a path derivation tree from $p$ to $p$.

\begin{lemma}\label{derive_subruns_easyway}
If $T$ is a derivation tree from $p$,
then $\under(T)$ is a commutative subrun from $[p]$ to $F(T)$.
\end{lemma}

\begin{proof}
\begin{eqnarray*}
F(T) &=& \sum_{v \in V} \left(\target(v) - \sum_{w: P(w)=v} [\source(w)]\right) \\
&=& \sum_{v \in V} \target(v) - \sum_{v \in V - \set{v_0}} [\source(v)] \\
&=& \target(T) - \source(T) + [\source(v_0)] \\
&=& \target(U(T)) - \source(U(T)) + [{p}].
\end{eqnarray*}
We also have that $P(v) \ra_{\under(T)} v$ for each $v \in V$, thus $U(T)$ is indeed connected from $P(v_0) = p$.\myqed
\end{proof}

\begin{theorem}\label{derive_subruns}
If $p \in \States$, and
$R \in \bbN^\Trans$ is a commutative subrun from $[p]$ to $t\in\bbN^\States$,
then $R = \under(T)$ for some derivation tree $T$ from
$p$ such that $F(T) = t$.
\end{theorem}

\begin{proof}
We will construct the tree $T$ inductively.
We start with $T_0 = (V, v_0, P, \transf)$,
where $\transf(v_0)$ is $\delta$ from Theorem \ref{subrun_lemma} for $R$, and $V = v_0$.
We will keep the following invariant: $U(T) \leq R$, and $R-U(T)$ is a subrun from
$F(T)$ to $t$. From Theorem \ref{subrun_lemma} the invariant is satisfied for $T_0$.
Suppose that the invariant is also satisfied for $T$, and $U(T) < R$.
. From Theorem $\ref{subrun_lemma}$ again
there exists a $\delta' \in R-U(T)$ such that $\source(\delta') \in F(T)$. Since 
$\source(\delta') \in F(T)$, there exists a $v$ such that $\source(\delta') \in F(v)$.
We obtain a new tree $T'$ by adding a new vertex $w$ such that $P(w) = v$ and 
$\transf(w) = \delta'$. From Theorem $\ref{subrun_lemma}$, the invariant is also satisfied
for $T'$.

Finally, we get a tree $T$ such that $U(T) = R$. From Lemma \ref{derive_subruns_easyway} we know that $F(T) = t$.
\myqed
\end{proof}

\begin{corol}\label{derive_runs_and_paths}
The following two conditions are equivalent:
\begin{enumerate}
\item $R \in \bbN^\Trans$ is a commutative run from $p$,
\item $R = \under(T)$ for some full derivation tree $T$ from $p$.
\end{enumerate}

Also, the following two conditions are equivalent:
\begin{enumerate}
\item $C \in \bbN^\Trans$ is a path from $p_1$ to $p_2$,
\item $C = \under(T)$ for some path derivation tree $T$ from
$p_1$ to $p_2$.
\end{enumerate}

In particular, the following two conditions are equivalent:
\begin{enumerate}
\item $C \in \bbN^\Trans$ is a cycle from $p$,
\item $C = \under(T)$ for some cyclic derivation tree $T$ from $p$.
\end{enumerate}
\end{corol}

\begin{proof}
Straightforward from Lemma \ref{derive_subruns_easyway} and
Theorem \ref{derive_subruns}. \myqed
\end{proof}

\begin{lemma}\label{gammalemma}
Let $G$ be a grammar in normal form.
Let $D$ be a derivation tree such that $D$ has no vertices at depth $M$.
Then $|D| < \gam M$, where $\gam M = M+1$ if $G$ is regular, and $2^{M+1}$ otherwise.
\end{lemma}

\begin{proof}
If each vertex has at most $C$ children, then there are at most
$C^d$ vertices at depth $d$. For regular grammars $C=1$, and for
grammars in normal form, $C=2$. All the hypotheses follow from
simple calculations. \myqed
\end{proof}

\begin{theorem}\label{cyclerunlimit}
Let $G$ be a grammar in normal form with $\qstate$ non-terminals. 
If $C$ is a simple cycle, then $|C| < \gam{\qstate}$.
If $R$ is a skeleton run, then $|R| < \gam{\qstate^2}$.
\end{theorem}

\begin{proof}
First, let $C$ be a simple cycle. From Corollary \ref{derive_runs_and_paths} we know
that $C = \under(T)$ for some cyclic derivation tree $T = (V,v_0,P,d)$ from $p$.

We will show that $T$ has no vertex at depth $N$.
Indeed, suppose otherwise that $T$ has a vertex $v$ at depth $n > N$. 
Consider a branch of the derivation tree: $v, p(v), p^2(v), \ldots, p^n(v) = v_0$.
Since $n \geq N$, we have $\source(p_i(v)) = \source(p_j(v))$ for some $i < j$. Let $V_1$ be all the
descendants of $p_i(v)$ (inclusive), $V_2$ be all the other descendants of $p_j(v)$, and $V_3$
be all the other vertices. Then $T_1 = (V_2, p_i(v), P, d)$ is a cyclic derivation tree,
and so is  $T_2 = (V_1 \cup V_3, p_i(v), P^*, d)$, where $P^*(p_i(v)) = P(p_j(v))$ and
$P^*(w) = P(w)$  for all other $w \in V_1 \cup V_3$. Hence, $C$ is not a simple cycle
($C = \under(T) = \under(T_1) + \under(T_2)$).

Now, let $R$ be a skeleton run. From Corollary \ref{derive_runs_and_paths} we know
that $R = \under(T)$ for some full derivation tree $T = (V,v_0,P,d)$ from $q_0$.

We will show that $T$ has no vertex at depth $N^2$.
Indeed, suppose otherwise that $T$ has a vertex $v$ at depth $n \geq \qstate^2$. 
Consider a branch of the derivation tree: $v, p(v), p^2(v), \ldots, p^n(v) = v_0$.
Since $n \geq \qstate^2$, there are indices $i_0, i_1, \ldots, i_N$ such that $\source(p_{i_k}(v))$
is the same non-terminal $q$ for each $k$. By repeating the construction above $N$ times,
we can decompose $R = R_0 + \sum_{k=1}^{\qstate} C_k$, where $R_0$ is a run and $C_k$ is the
cycle between $p_{i_{k-1}}(v)$ and $p_{i_k}(v)$. Let $U_K = \supp(R_0 + \sum_{k=1}^{K} C_k$. Since
$U_0 \subseteq U_1 \subseteq U_2 \ldots \subseteq U_N \subseteq \States$, $U_0$ is not empty, 
and $|\States|=N$, there must be $k$ such that $U_k = U_{k-1}$. Since $C_k$ does not add any
new non-terminals to the support, we have that $\supp(R-C_k) = \supp(R)$.
This contradicts the assumption that $R$ was a skeleton run ($R-C_k$ is a run).

All the hypotheses follow from the Lemma \ref{gammalemma}. \myqed
\end{proof}

\subsection{Compact Representation of a Commutative Regular Language}

In this subsection, we will show how to obtain a compact representation of a regular
commutative language over a fixed alphabet: such a language is a union of polynomially
many polynomially bounded simple bundles. This compact representation will be used in 
Section \ref{seccomplexity} below to prove that the membership problem is in P for
regular commutative grammars over a fixed terminal language. We also obtain a
less compact representation in the non-regular case.

\begin{theorem}\label{decomposition}
Let $G$ be a grammar in normal form with $N$ non-terminals over an alphabet
of size $\alphs$, and 
let $R$ be a run. Let $\frunbound = \gam{\qstate^2} + (2\gam{\qstate})^{1+\alphs} \hadamard\alphs{\gam{\qstate}}$.
Then $\out(R) = \out(R_1) + \sum \out(C_k) n_k$, where:
\begin{itemize}
\item $R_1$ is a run such that $|R_1| \leq \frunbound$,
\item each $C_k$ is a simple cycle from some $q \in \supp(R_1)$,
\item $\out(C_k)$ are linearly independent,
\item $n_k \in \bbN$.
\end{itemize}
\end{theorem}

\begin{proof}
Let $R$ be a run. As long as $R$ is not a skeleton run, we can decompose $R$ as 
a sum of a smaller run, and a simple cycle. Thus, we obtain that $\out(R) = \out(R_0) + \sum_{k \in K} \out(C_k) n_k$.

Suppose that $\out(C_k) = \out(C_l)$ for some $k \neq l$. Then we remove $l$ from $K$, and 
add $n_l$ to $n_k$. The equation $\out(R) = \out(R_0) + \sum_{k \in K} \out(C_k) n_k$ still holds.

Let $K^* = \set{k \in K: n_k \geq \hadamard\alphs{\gam \qstate}}$. 
From Lemma \ref{detrewritemuch} we can assume that $\out(C_k)$ are linearly
independent for $k \in K^*$.


Since for each $k$ we have $\nsum{C_k} < \gam N$, and $\out(C_k)$ is different
for each $k$, we have at most $(2\gam \qstate)^\alphs$ cycles there. 
Let $R_0 = R_1 + \sum_{k \in {K-K^*}} C_k n_k$. We have $|R_0| \leq |R_1| + 
(2\gam \qstate)^{\alphs+1} \hadamard\alphs{\gam \qstate}$, and $\out(R) = \out(R_0) + \sum_{k \in K^*} \out(C_k) n_k$.
\myqed
\end{proof}

\begin{corol}\label{regularbundle}
If $G$ is a regular grammar with $N$ non-terminals over an alphabet of size $\alphs$, then
$\Parikh(G)$ is a union of at most $N^{\alphs^2}$ simple bundles bounded by $(\frunbound,\qstate)$.
\end{corol}

Also note that $\frunbound$ is polynomial for regular grammars over a fixed alphabet.

\begin{proof}
From Theorem \ref{decomposition} we know that for each run $R$ we have
$\out(R) = \out(R_1) + \sum \out(C_k) n_k$, where $\out(R_1)$ is bounded by $\frunbound$ and
$C_k$ is a simple cycle. We bundle the runs which use the same cycles together.
Since simple cycles are bounded by $N$, there are at most $N^\alphs$ of them, and there
are at most $(\qstate^\alphs)^\alphs$ sets of linearly independent simple cycles.\myqed
\end{proof}

It is possible to get a better bound on the number of simple bundles in dimension 2, even for non-regular grammars.

\begin{corol}\label{dimtwobundle}
If $G$ is a grammar in normal form with $N$ non-terminals over a two-letter alphabet, then
$\Parikh(G)$ is a union of $O(\qstate^2)$ simple bundles bounded by $(\frunbound,\gam{\qstate})$.
\end{corol}

\begin{proof}
Let $\Alphabet = \{a_1, a_2\}$. First assume that the grammar is positive.

For each non-terminal $q$, among all the cycles from $q$, let $C_i(q)$ be
the one with the greatest proportional amount of $a_i$, for $i=1,2$. All the other
cycles from $q$ fall in the angle between $C_1(q)$ and $C_2(q)$.

Now, for each run $R$, let $C_i(R)$ be the one with the greatest proportional amount
of $a_i$, among all cycles from $q \in \supp(R)$.

Proceed as in the proof of Corollary \ref{regularbundle}, except now we can assume that
the cycles $C_1$, $C_2$ are always $C_1(R)$ and $C_2(R)$ (instead of arbitrary cycles),
as all other cycles can be written as positive linear combinations of these two.
Since $C_1(R)$ and $C_2(R)$ are chosen from $C_i(q)$, there at at most $N^2$ simple
bundles.

For non-positive grammars, for each non-terminal $q$, one of the two cases holds:
\begin{itemize}
\item all cycles from $q$ (and their positive combinations) fall in the angle between
the two extreme cycles $C_1(q)$ and $C_2(q)$,
\item positive linear combinations of cycles from $q$ cover the full plane, and can be written
as positive combinations of $C_1(q)$, $C_2(q)$ and $C_3(q)$ which are cycles from $q$.
\end{itemize}

\noindent Thus, we can always use one of the at most three cycles $C_i(q)$ for some $q \in \supp(R)$.
There are still $O(\qstate^2)$ combinations of them.\myqed
\end{proof}

\section{Window Theorem for Commutative Grammars}\label{seccommg}

This section is devoted to the following result, which we call the \emph{window theorem}: 
in order to determine whether two commutative languages $\Parikh(G_1)$ 
and $\Parikh(G_2)$ over a fixed alphabet $\Alphabet$ of size $A$ defined by
grammars $G_1$ and $G_2$ with $N$ non-terminals are disjoint or
equal, it is enough to only examine a small window of size which is
single exponential in $N$. This result will be instrumental in the proof that
inclusion, universality, and disjointness problems are in \PiP\ for commutative
grammars over an alphabet of fixed size (Theorem \ref{inclcf} below).

The situation is much easier for $\alphs=2$ than in the general case. From Corollary
\ref{dimtwobundle}
we know that each of $\Parikh(G_1)$ and $\Parikh(G_2)$ is a union of $O(N^2)$ simple bundles
bounded by $(\frunbound,\gam{\qstate})$. In this case, we can show the following result:
for each $v \in \bbZ^\Alphabet$ we can find a $v_0$ of single exponential magnitude
such that $v_0$ is in exactly the same of our $O(N^2)$ bundles as $v$. This is achieved
in Lemma \ref{secbound} in the following way:
\begin{itemize}
\item Let $W + \osum\bbN \Cout$ be one of the bundles. 
From Lemma \ref{detgroup} we know that each member of $(\det M) \bbZ^\Alphabet$,
where $M$ is the matrix whose columns are $\Cout$, is a member of $\osum\bbZ \Cout$.
Hence, if $w \in (\det M) \bbZ^\Alphabet$, then $v \in W + \osum\bbZ \Cout$ iff
$v+w \in W + \osum\bbZ \Cout$. Since
there is just a polynomial number of bundles, we can use the least common multiple
of the determinants to ensure that the equivalence above is satisfied
for each bundle.
\item The bundle is $W + \osum\bbN \Cout$, not $v \in W + \osum\bbZ \Cout$, thus
we need to ensure that the signum of the coefficients remains unchanged.
This is done by partitioning $\bbR^\Alphabet$ into \emph{regions}. Two points $v$, $v_0$
are in the same region if $v \in w + \osum\bbP\Cout$ iff $v_0 \in w + \osum\bbP\Cout$
for each $w, \Cout$ satisfying the necessary bounds. We show that we can additionally
ensure that the single exponentially bounded $v_0$ is in the same region as $v$,
which proves the lemma.
\end{itemize}

\noindent However, this approach no longer works for $\alphs>2$, as $\Parikh(G_i)$ no longer
needs to be a union of polynomial number of simple bundles; section \ref{hardgrammar}
below is devoted to showing an example of such a grammar.

We solve this problem by using the regions again. Although $\Parikh(G_i)$ need not
be a union of a polynomial number of simple bundles, this is true when we consider
regions separately: for each region $r$, $\Parikh(G_i) \cap r$ equals $U \cap r$, where
$U$ is a union of polynomial number of simple bundles. This is proven 
in Lemma \ref{mainbound} below.

The rest of this section provides a detailed statement and proof of our window theorem.
In fact, we will prove the following Lemma about semilinear sets, without
directly using the assumption that our semiliner sets come from commutative grammars;
the window theorem about commutative grammars (Theorem \ref{normalcommcorol} below)
will follow easily from it.

\begin{lemma}\label{normalcommlemma}
Let $\frunboundx \in \bbN$.
Let $S_1$ and $S_2$ be two semilinear sets over the same fixed alphabet $\Alphabet$
of size $\alphs$, given as $S_k = \bigcup_{i \in I_k} W_i + \osum\bbN{\calZ_i}$, where
$W_i\subseteq [-\frunboundx..\frunboundx]^\Alphabet$,
$\calZ_i\subseteq [-Y..Y]^\Alphabet$, for every $i \in I_1 \cup I_2$.
Then there exists a number
$\grammarboundl = O((\frunboundx+Y)^{|I_1 \cup I_2|})$ such that
$S_1 \subseteq S_2$ iff          
$S_1 \cap [-\grammarboundl..\grammarboundl]^\Alphabet \subseteq S_2 \cap [-\grammarboundl..\grammarboundl]^\Alphabet$,
and
$S_1$ is disjoint with $S_2$ iff
$S_1 \cap [-\grammarboundl..\grammarboundl]^\Alphabet$ is disjoint with $S_2 \cap [-\grammarboundl..\grammarboundl]^\Alphabet$.
\end{lemma}

\begin{theorem}\label{normalcommcorol}
Let $G_1$ and $G_2$ be two commutative grammars in normal form 
with at most $N$ non-terminals each,
over the same fixed alphabet $\Alphabet$ of size $\alphs$. Then there exists a number
$\grammarbound$ which is single exponential in $N$, such that
$\Parikh(G_1) \subseteq \Parikh(G_2)$ iff
$\Parikh(G_1) \cap [-\grammarbound..\grammarbound]^\Alphabet \subseteq \Parikh(G_2) \cap [-\grammarbound..\grammarbound]^\Alphabet$,
and
$\Parikh(G_1)$ is disjoint with $\Parikh(G_2)$ iff
$\Parikh(G_1) \cap [-\grammarbound..\grammarbound]^\Alphabet$ is disjoint with $\Parikh(G_2) \cap [-\grammarbound..\grammarbound]^\Alphabet$.
\end{theorem}

\begin{proof}[Proof of Theorem \ref{normalcommcorol}]
Let $G_u = (\Alphabet, \States_u, q_u, \Trans_u)$. Without loss of generality we can assume
that the sets $\States_u$ are disjoint for $u \in \set{1,2}$; this way, we will be able to
identify the grammar of each run and cycle by mentioning one of the non-terminals used.

For $q \in \States_u$, let $\Cout_q = \set{\out(C): C\mbox{ is a simple cycle from }q}$,
and for $S \subseteq \States_u$, let $\Cout_S = \bigcup_{q \in S} \Cout_q$.

From Theorem \ref{decomposition} we know that 
$v \in \Parikh(G_u)$ iff there exists a run $R_1$ from $q_u$ such that $|R_1|$ is bounded
by $\frunbound$ and a linearly independent set of simple cycle outputs
$P \subseteq \Cout_{\supp{R_1}}$ such that $v \in \out(R_1) + \osum\bbN P$.

In particular, we know that $P$ is of cardinality at most $\alphs$. Let $I_u$ be the set of
all subsets of $\States_u$ of cardinality at most $\alphs$. Thus, we know that $v \in \Parikh(G_u)$
iff there exists $J \in I_u$, a run $R_1$ of $G_u$ such that $J \subseteq \supp(R_1)$
and $|R_1|$ is bounded by $\frunbound$, and a set $P \subseteq \Cout_J$ such that $v \in \out(R_1) + \osum\bbN P$.
Let $W_J = \set{\out(R): R\mbox{ is a run in $G_u$ bounded by $\frunbound$ such that $J \subseteq \supp(R)$}}$. Thus, we get
that $\Parikh(G_u) = \bigcup_{J \in I_u} (W_J + \osum\bbN \Cout_J)$. We know that $W_J$
is bounded by $\frunbound$ which is single exponential, and $\Cout_J$ is bounded by $Y=\gam \qstate,$
which is also single exponential. 

Since the cardinality of $I_u$ is polynomial in $N$, we get our claim by applying Lemma \ref{normalcommlemma}. \myqed
\end{proof}

\begin{proof}[Proof of Lemma \ref{normalcommlemma}]
We will need to introduce the notion of a \emph{region}.

Let $\calS_Y$ be the family of all linearly independent subsets of $[-Y..Y]^\Alphabet$
containing $\alphs-1$ elements. For each element $S$ of $\calS_Y$ and $v \in \bbR^\Alphabet$, let
$\phi_S(v)$ be the determinant of the matrix $M_{S,v}$ whose $\alphs-1$ columns are the
elements of $S$ (ordered in an arbitrary way), and the last column is $v$. From the
properties of the determinant, we know that $\phi_S(v) = \sum_{a \in \Alphabet} \alpha_a v_a$,
where $\alpha_a$ is the (possibly negated) determinant of the $(\alphs-1) \times (\alphs-1)$
submatrix  of $M_{S,v}$ which misses the column $v$ and the row $a$. In particular,
$\alpha_a \in [-\phibound..\phibound]$, where $\phibound = \hadamard{\alphs-1}{Y}$. Intuitively, for each
$S$, $\phi_S^{-1}(0)$ is the $\alphs-1$-dimensional subspace containing all elements
of $S$. By calculating $\phi_S(v)$ we can tell whether $v$ is above or below this
subspace. Let $\Phi_Y = \set{\phi_S: S \in \calS_Y}$. The cardinality of $\Phi_Y$ is bounded by
$(2Y+1)^{\alphs^2}$, and its elements have coefficients bounded by $\phibound$. In particular,
if $\phi \in \Phi_Y$, then $|\phi(v)| \leq \phibound \nsum{v}$, and $\nmax{\phi(v)} \leq \aphibound \nmax{v}$.

Now, let $L_B = [-2B, 2B]^\Alphabet$, and $\Regions_{B,Y}$ be the set of all functions from
$\Phi_Y \times L_B$ to $\{-1,0,1\}$.

For $r \in \Regions_{B,Y}$, we define strict regions $\Reg(r)$, weak regions $\reg(r)$,
and ray regions $\tau(r)$ as follows:

\begin{eqnarray*}
\Reg(r) &=& \set{x \in \bbR^\Sigma: \forall \phi \in \Phi_Y\ \forall l \in L_B\ \sgn(\phi(x-l)) = r_{\phi,l}}      \\
\reg(r) &=& \set{x \in \bbR^\Sigma: \forall \phi \in \Phi_Y\ \forall l \in L_B\ \sgn(\phi(x-l))\ r_{\phi,l} \geq 0} \\
\tau(r) &=& \set{x \in \bbR^\Sigma: \forall \phi \in \Phi_Y\ \forall l \in L_B\ \sgn(\phi(x))\ r_{\phi,l} \geq 0}
\end{eqnarray*}

The following picture shows what is going on for $\alphs=2$ and $Y=3$.

\begin{center}
\label{figure2d}
\includegraphics[width=8cm]{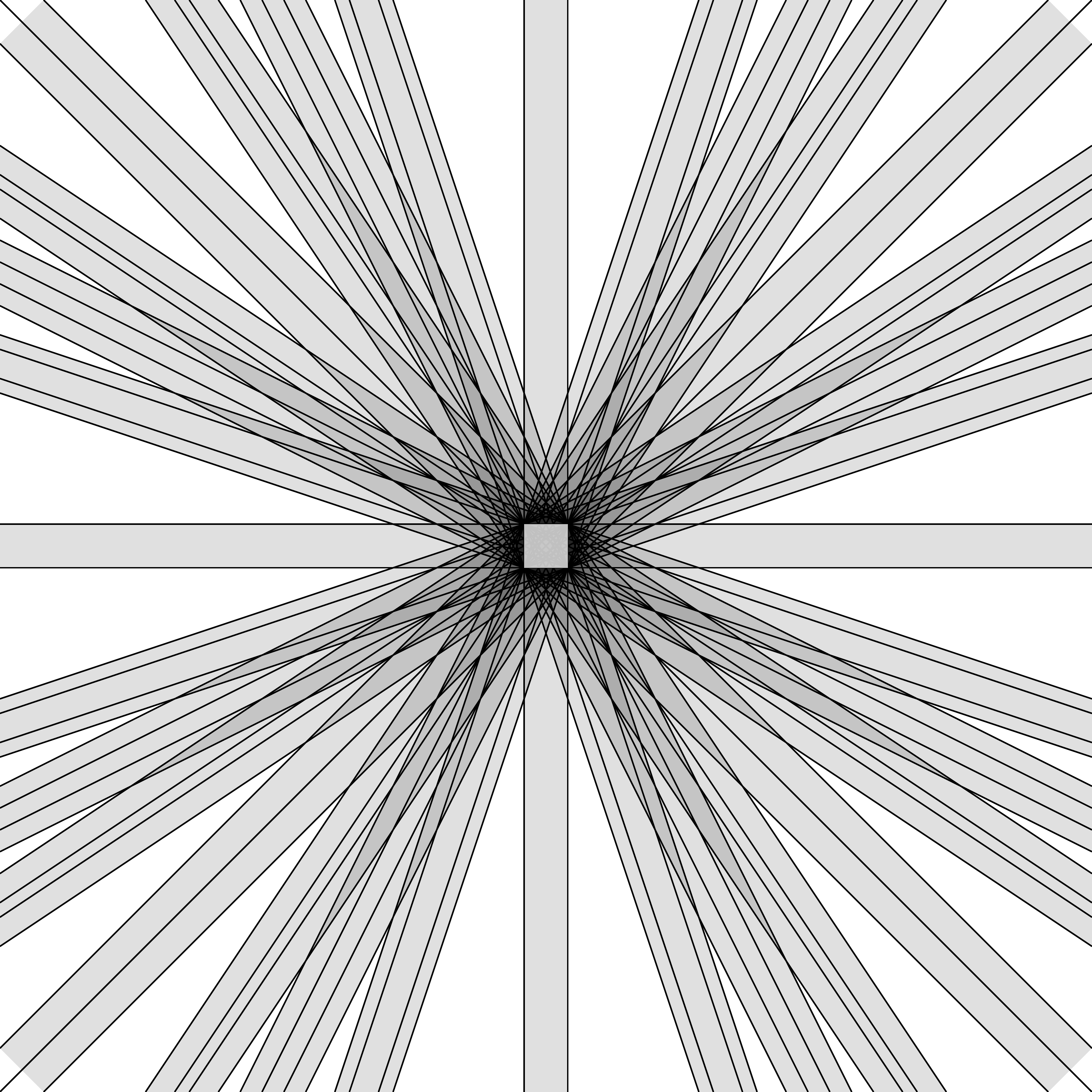}
\end{center}

The gray square in the center is $L_B$. From each point of the square, and for
each tuple of $\alphs-1$ vectors in $[-Y..Y]^\Alphabet$, we shoot a
hyperplane which is parallel to each of these vectors. For $\alphs=2$, this means
that we shoot a line in each direction given by some vector in $[-Y..Y]^\Alphabet$. 
There are 32 directions, and for each of them we get a bundle of lines, by taking
different points of $L_B$ to shoot from; such bundles of lines are marked gray on
the picture above.
Each strict region $\Reg(r)$ is given by its relation to each of these lines (above, below,
or on one of these lines). There are four types of regions:

\begin{itemize}
\item empty regions (the relations are inconsistent),
\item a single point ($r$ says that the point is on two (non-parallel) lines at once),
\item a semiline ($r$ says that the point is on a line and above some other line), 
\item a wedge ($r$ says that the point is below line $l_1$ and above line $l_2$, where
$l_1$ and $l_2$ are not parallel).
\end{itemize}

\noindent There are 32 wedge-shaped regions, an infinite number of semiline-shaped regions
(bundled into 32 packs of parallel semilines), and an infinite number of points.

Strict regions $\Reg(r)$ are disjoint and partition the plane, while the weak regions
$\reg(r)$ are their closures.
Ray regions $\tau(r)$ look similar to $\Reg(r)$, but we shoot lines only from 0, instead 
of each point from $L_B$. There are 32 wedge-shaped and 32
semiline-shaped regions,
each of them start at 0 (which is included into each $\tau(r)$). If $\Reg(r)$ is a single
point or the empty set, then $\tau(r) = \{0\}$.

The situation is more complicated in three dimensions, and it is also harder to draw.
We will show how $\tau(r)$ looks for $\alphs=3$ and $Y=2$. From the definition of
$\tau(r)$ we know that for each $\alpha>0$ we have $x \in \tau(r)$ iff
$\alpha x \in \tau(r)$. Thus, it is sufficient to draw only a situation on a planar
section of $\bbR^\Alphabet$.
Let $\Delta = \{v \in \bbR^\Alphabet: v \geq 0, \nsum v = 1\}$. The set $\Delta$ is a
triangle; this triangle is shown on the picture below.

\begin{center} \label{figure3d}
\includegraphics{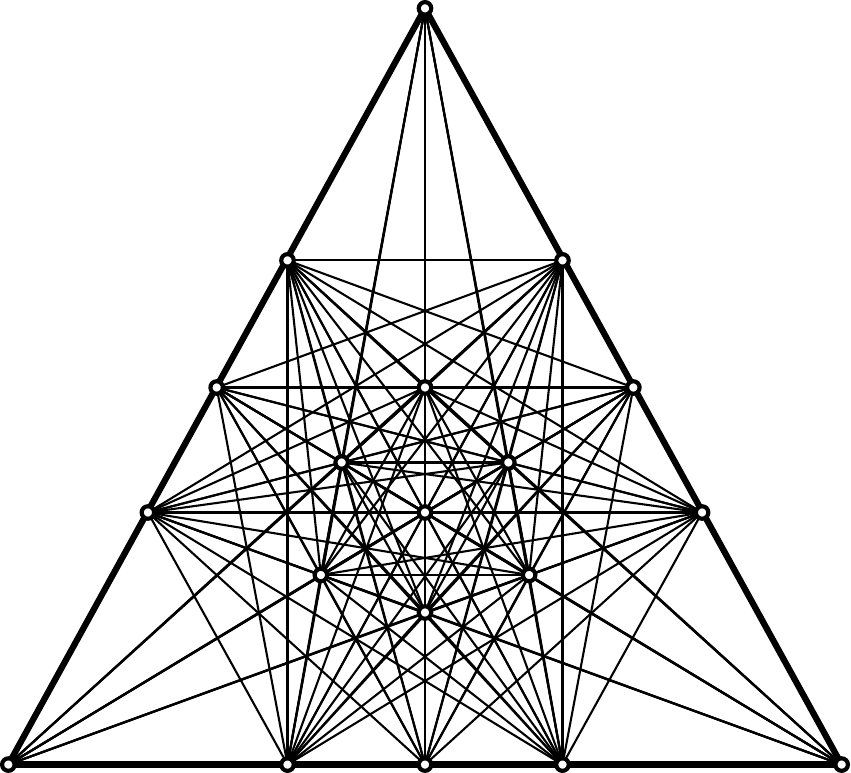}
\end{center}

The 19 white dots on the picture denote elements of
$[0..Y]^\Alphabet$; for each such $v$, its projection $v \over \nsum{v}$ is shown 
as a white dot on the triangle (in fact, we should also have white dots for
elements of $[-Y..Y]^\Alphabet$ with negative coordinates, but this would obfuscate the
picture too much --- their projections on our planar section do not fit in the triangle).
$\calS$ is the set of pairs of linearly independent elements of $\zmint{Y}^\Alphabet$,
or pairs of distinct white dots. For each $S \in \calS$, we have the linear function
$\phi_S$, which equals 0 on both elements of $\calS$. These are represented by the black
lines: for each two white dots, we have a black line going through these white dots. 

Now, let $r \in \Regions_{B,Y}$. The set $\tau(r)$ is the set of points which are on
the side of the semispace $\phi$ declared by $r_{\phi,l}$, for each $\phi$ and $l$
(or on this semispace, in case if $r_{\phi,l} =0$). In our case, semispaces are
the black lines. Therefore, $\tau(r) \cap \Delta$ could be either an
empty set (if this is inconsistent -- this happens for example when $r_{\phi,l_1} \neq 
r_{\phi,l_2}$), or a region of the triangle bounded by the black
lines, or one of the white dots, or one of the points where the black lines cross,
or a segment of a black line between two points where it crosses the other lines.
The set $\tau(r)$ can be described by multiplying $\tau(r) \cap \Delta$ by each scalar
$\alpha \geq 0$. Thus, each black line becomes a black plane, each white dot or black
line crossing becomes a semiline, each black line segment becomes an infinite triangle,
and each bounded region becomes an infinite cone.

Regions $\reg(r)$ and $\Reg(r)$ are harder to visualize --- we have to replace each
black plane by a pack of black planes by moving it by each vector in $L_B$, and again
split $\bbR^\Alphabet$ by these new black planes.

The following lemma describes the shape of the regions.

\begin{lemma}\label{regionshape}
Let $r \in \Regions_{B,Y}$. Then:

\begin{itemize}
\item $\tau(r)$ can be written as $\osum\bbP T_r$, where $T_r \subseteq [-\taulimit..\taulimit]^\Alphabet$,

\item $\reg(r) = W_r + \tau(r)$, where $W_r \subseteq [-\regbaselimit,\regbaselimit]^\Alphabet$.
\end{itemize}

In particular, if $\reg(r)$ is bounded, then $\tau(r) = \{0\}$.
$\taulimit$ is bounded polynomially by $Y$, and $\regbaselimit$ is bounded polynomially by $B$ and $Y$.
\end{lemma}

For example, let us consider the two-dimensional case from the picture on page
\pageref{figure2d}. Take $r$ such that $\reg(r)$ is an infinite wedge with a vertex
in some point $w$. Then
$\tau(r)$ is a congruent wedge with a vertex in point 0, bounded by semilines parallel
to some two vectors $v_1, v_2 \in [-Y..Y]^\Alphabet$. $T_r$ will be simply $\{v_1, v_2\}$,
and $W_r$ will be $\{w\}$. Our lemma says that $w$ is bounded polynomially in $B$ and $Y$;
since the picture is big enough to already show 
all the 32 wedge vertices, we can take $\regbaselimit$ to
be (half of) the edge of the square shown on the picture.
In the three-dimensional case the situation will be more
difficult: elements ot $T_r$ will be bounded by a polynomial in $Y$ (instead of $Y$
itself), $W_r$ will no longer have to be a singleton, and the cardinality of set $T_r$
might be relatively large -- for example, 
on the picture on page 
\pageref{figure3d} there are ray regions $\tau(r) \cap \Delta$ bounded by five black
lines, so we will need a $T_r$ with five elements to describe the respective $\tau(r)$.

Having introduced the regions, we can use the following lemma:

\begin{lemma}\label{mainbound}
Let $B, Y \in \bbN$. Then there is a $\mainboundlimit$, bounded polynomially in $B$ and $Y$,
such that the following holds:

Let $S = W + \osum\bbN \Cout$, where $\Cout \subseteq \zmint{Y}^\Alphabet$ and 
$W \subseteq [-B..B]^\Alphabet$, and $r \in \Regions_{B,Y}$. Then there is a linearly
independent subset $\Cout_0 \subseteq \Cout$ and $W_1 \subseteq \zmint{\mainboundlimit}^\Alphabet$
such that $S \cap \reg(r) = (W_1 + \osum\bbN \Cout_0) \cap \reg(r)$.
\end{lemma}

This lemma says that, if we restrict ourselves to a single region, 
we can improve our presentation $S_u = \bigcup_{i \in I_u} (W_i + \osum\bbN \Cout_i)$
by making the sets $\Cout_i$ linearly independent, while still keeping the bound
on $W_i$ exponential. Thus, in each region, $S_u$ is a union
of $|I_u|$ simple bundles bounded by $(\mainboundlimit,Y)$.

Now, we use the following lemma:

\begin{lemma}\label{secbound}
Let $\Alphabet$ be a fixed alphabet, $\mainboundlimit, Y \in \bbN$, and $I$ be a set of indices.
Then there exists a number $\secv = O((\mainboundlimit+Y)^{|I|})$ such that the following holds:

Let $V_i = W_i + \osum\bbN \Cout_i$ be a simple bundle bounded by $(\mainboundlimit,Y)$, 
for each $i \in I$, and $r \in \Regions_{\mainboundlimit,Y}$. Then for each
$v \in \bbZ^\Alphabet \cap \Reg(r)$ there exists a
$v_0 \in [-\secv,\secv]^\Alphabet \cap \Reg(r)$ such that for each
$i \in I$ we have $v \in V_i$ iff $v_0 \in V_i$.
\end{lemma}

By applying the lemma above to $I = I_1 \cup I_2$ and taking $\grammarboundl = \secv$,
we obtain our hypothesis. \myqed
\end{proof}

\begin{proof}[Proof of Lemma \ref{regionshape}]



Without loss of generality we can assume that $\Reg(r) \geq 0$.
Indeed, for each $a \in \alphs$, The set $\Phi_Y$ includes the linear function $\phi_a$ given by
$\phi_a(x) = x_a$. If $r_{\phi,0} \geq 0$, then $x_a \geq 0$ for all $x \in \Reg(r)$;
similarly if $r_{\phi,0} \leq 0$. In the second case, we can take the mirror image,
which is also a region.

Take $\phi \in \Phi_Y$. We can write the definitions of $\reg(r)$ and $\tau(r)$
as follows: $\reg(r) = \bigcap_\phi H^\reg_\phi$, $\tau(r) = \bigcap_\phi H^\tau_\phi$, where
\[H^\reg_\phi = \bigcap_{l \in L_B} \set{x \in \bbR^\Sigma:\ \sgn(\phi(x-l))\ r_{\phi,l} \geq 0},\]
\[H^\tau_\phi = \bigcap_{l \in L_B} \set{x \in \bbR^\Sigma:\ \sgn(\phi(x))\ r_{\phi,l} \geq 0}.\]
The sets $H^*_\phi$ are half-spaces or closed hyperplanes; $H^\tau_\phi$ (or its boundary) goes through 0,
and $H^\reg_\phi$ is parallel.

Let $\sigma(v) = \sum_i v_i$, and $\Delta = \set{v \geq 0: \sigma(v) = 1}$. The set
$\tau(r)$ is described by $\tau(r) \cap \Delta$: for $x \in \Delta$ and $\alpha>0$,
we have $\alpha x \in \tau(r)$ iff $x \in \tau(r) \cap \Delta$. The set $\Delta$ is
a bounded $\alphs-1$-dimensional polytope (simplex), and $\tau(r) \cap \Delta$ is
its intersection with a finite number of closed halfspaces and hyperplanes
($H^\tau_\phi$). Thus, $\tau(r) \cap \Delta$ is also
a bounded polytope. Each vertex of this polytope is the point where $A-1$ hyperplanes
of form $\phi(v)=0$ (where $\phi \in \Phi$) cross $\Delta$. The coefficients
of $\phi$ are bounded
polynomially by $Y$, so each vertex can be written as $t'_i = t_i / \sigma(t_i)$,
where $\sigma(t_i)$ is bounded polynomially by some $\taulimit$ (by Lemma \ref{syseq}). 
Let
$T_r = \{t_1, t_2, \ldots, t_N\}$. Note that the number of vertices, $N$, is bounded
polynomially.

It is easy to check that the set $\reg(r)$ is \emph{upward closed}, i.e., if
$x \in \reg(r)$, $i \in [1..N]$, and $\alpha \geq 0$, then $x+\alpha t_i \in \reg(r)$.

Let $W_r$ be the set of points in $\reg(r)$ such that for all $\alpha > 0$ and $i$,
$W_r-\alpha t_i \notin \reg(r)$. We have that $\reg(r) = W_r + \osum\bbP T_r$.
Indeed, $\supseteq$ follows from upward closedness, and for $\subseteq$,
take $x_0 \in \reg(r)$. For $i$ from
1 to $N$, let $x_i = x_{i-1} - \alpha_i t_i$, where $\alpha_i \in \bbP$ is the greatest 
number such that $x_{i-1} - \alpha_i t_i$ is still in $\reg(r)$ -- such an $\alpha$ must
exist, since $\reg(r)$ is closed, $\reg(r) \geq 0$, and $x_{i-1} - \alpha t_i$ has
a negative coordinate for large enough values of $\alpha$. We obtain a point $x_N$,
which must be in $W_r$. Indeed, suppose that $x' = x_N-\alpha t_i \in \reg(r)$. In that
case, $x_{i-1} - (\alpha_i + \alpha) t_i = x' + \sum_{j={i+1}}^{N} \alpha_j t_j\in \reg(r)$
because $x' \in \reg(r)$ and $\reg(r)$ is upward closed. This contradicts the assumption
that the value of $\alpha_i$ was maximum.

Take $x \in W_r$. Let $H^b_\phi = H^\reg_\phi$ in case if $H^\reg_\phi$ is a hyperplane,
or the boundary hyperplane of $H^\reg_\phi$ if $H^\reg_\phi$ is a half-space.
Let $\Psi_x \subseteq \Phi_Y = \set{\phi: x \in H^b_\phi}$. The set
$H_x := \bigcap_{\phi \in \Psi_x} H^b_\phi$ is an intersection of hyperplanes, and thus
it is a subspace.

The set $H_x \cap W_r$ has to be bounded. Indeed, suppose that $H_x \cap W_r$
is not bounded. Take a sequence $(v_i)$ of elements of $H_x \cap W_r$ such that
$\lim_{i\ra\infty} \sigma(v_i) = \infty$. Let $w_i = v_i / \sigma(v_i)$; we have 
$w_i \in \Delta$. Since $\Delta$ is compact, $w_i$ has a cluster point, $w$. 
For $\phi \in \Phi_Y$ and $l \in L_B$, we have 
$\phi(w) = \phi(w_i-l)/\sigma(w_i) + \phi(l)/\sigma(w_i)$. The first component is
zero or has sign $r(\phi,l)$ (since $v_i \in \reg(r)$)
and the second component tends to 0, thus $\phi(w)$ is either 0, or it has
sign $r(\phi,l)$. Thus, $w \in \tau(r)$, and also $w \in H_x-H_x$ (the subspace 
parallel to $H_x$ going through 0). The set $H_x \cap W_r$ is a polytope which is
unbounded in the direction of $w$, but this is impossible, since $x$ and $x+\alpha w$
cannot be both in $W_r$ (from the definition of $W_r$). Hence, a contradiction.

Moreover, the set $H_x \cap W_r$ has to be bounded polynomially. Indeed, the vertices
of $H_x \cap W_r$ are points where some $\alphs$ hyperplanes of form $H^\reg_\phi$
cross. These hyperplanes have coefficients bounded polynomially, and from Lemma \ref{syseq}
we get that the coordinates of vertices of $H_x \cap W_r$ are bounded polynomially.

Thus, $W_r = \bigcup_{x \in W_r} (H_x \cap W_r)$ is bounded polynomially.\myqed
\end{proof}

\begin{proof}[Proof of Lemma \ref{mainbound}]
If $\Reg(r)$ is bounded, then it is bounded polynomially in $B$ and $Y$ (Lemma \ref{regionshape}),
and therefore we are done. Thus, we can assume that $\Reg(r)$ is unbounded.

Suppose that there is a $t \in \tau(r)$ such that $t \notin \osum\bbP\Cout$.
We will show that in this case $\Reg(r)$ and $S = W + \osum\bbP\Cout$ are disjoint, and therefore we are done.

Indeed, note that $\osum\bbP\Cout$ is a $d$-dimensional cone (for some $d$), whose faces are
hyperplanes going through 0 and parallel to sets of $(d-1)$ vectors in $[-Y,Y]^\Alphabet$.
The ray region $\tau(r)$ cannot be subdivided into two smaller regions by such a hyperplane,
therefore there is one of these hyperplanes, say $\phi$, such that $\phi(t)<0$,
$\phi(\tau(r)) \leq 0$, and $\phi(\Cout) \geq 0$. 
Take $x \in \Reg(r)$ such that $x = w + \sum \alpha_i \Cout_i$, where
$w \in [-B..B]^\Alphabet$. 
Since $\phi(t) < 0$, for each $l \in L_B$
we have $r_{\phi,l} < 0$, and thus, since
$x \in \Reg(r)$, we have
$\phi(x-l) \leq 0$.
Since $L_B = [-2B,2B]^\Alphabet$, there is a $l_0 \in L_B$ such that $\phi(l_0) < \phi(w)$.
Thus, $\phi(x) \leq \phi(l_0) < \phi(w) \leq \phi(w) + \sum \alpha_i \phi(\Cout_i)
= \phi(x)$, which is a contradiction.

Thus, we can assume that $\tau(r) \subseteq \osum\bbP\Cout$. Each point in 
$\osum\bbP\Cout$ can be written as $\sum \alpha_i \Cout_i$, where $\alpha_i$ is positive
only for a linearly independent set of vectors in $\Cout$ 
(otherwise, if $\alpha_i$ are positive in a linearly dependent subset of $\Cout$,
we can modify the values so that one of them is 0 -- just like in Lemma \ref{detrewritemuch}).
Therefore,
$\tau(r) \subseteq \bigcup_{F \in \calF} \osum\bbP F,$ where $\calF$ is the family
of all linearly independent subsets of $\Cout$. Again, each 
$\osum\bbP F$ is a $d$-dimensional cone, whose faces are
hyperplanes going through 0 and parallel to sets of $(d-1)$ vectors in $[-Y,Y]^\alphs$.
The ray region $\tau(r)$ cannot be subdivided into two smaller regions by such a hyperplane,
thus there is a $F\in\calF$ such that $\tau(r) \subseteq \osum\bbP F$. Let 
$\Cout_0$ be this $F$, and let $D$ be determinant of the matrix whose columns are
$\Cout_0$ (we add unit vectors if there are less than $\alphs$ vectors in $\Cout_0$).
From Lemma \ref{smalldet} we know that $D \leq \hadamard\alphs{Y}$, which is single
exponential.

Now, let $W_1 = S \cap \zmint{\mainboundlimit}^\Alphabet$; the sufficient value of $\mainboundlimit$ will be
apparent from the sequel of the proof.
By induction over $\nsum{v}$, we will prove that each $v \in S$ can be written as
$v = w_1 + \osum\bbN\Cout_0$, where $w_1 \in W_1$.

For $v \in S$ such that $\nsum{v} \leq \mainboundlimit$, the hypothesis obviously holds.

For $v \in S$ such that $\nsum{v} \geq \mainboundlimit$ we apply the following Lemma:

\begin{lemma}\label{reducerlemma}
Let $\Alphabet$ be a fixed alphabet, and $B, Y, D \in \bbN$. Then there exists a constant
$\reduceto$ polynomial in $B$, $Y$ and $D$, such that for each
$r \in \Regions_{B,Y}$ and each $v \in \reg(r)$, if $\nmax{v} \geq \reduceto$, then
$v = v_0 + Dt$, where $v_0 \in \reg(r)$ and $t \in \tau(r) \cap [-\taulimit..\taulimit]^\Alphabet$
(where $\taulimit$ is from Lemma \ref{regionshape}).
Moreover, $v_0 \in \Reg(r)$ iff $v \in \Reg(r)$.
\end{lemma}

Take $v \in S$. By applying Lemma \ref{reducerlemma} iteratively, we know that
$v = v_0 + D(t_1 + \ldots + t_K)$, where $t_i \in \tau(r) \cap [-\taulimit..\taulimit]^\Alphabet$, and
$\nmax{v_0} \leq \reduceto$.

On the other hand, we know that $v = w_0 + w_{12}$, where $w_0 \in W$ and
$w_{12} = \sum_{v \in \Cout} \beta_v v$, where $\beta_v \in \bbN$. 
From Lemma \ref{detrewritemuch} we can assume that $\beta_v$ exceed
$H = \hadamard\alphs{Y}$ only in a linearly independent $P \subseteq \Cout$.
Take a large enough $J_0 \geq H$, polynomial in B and Y (the required value will be
apparent below).
We can write $v = w_0 + w_1 + w_2$, where 
$w_1 \in \osum{[0..J_0]}\Cout$ and $w_2 = \sum_{v \in P} \beta_v v$, where
each $\beta_v \geq J_0$, and $P$ is linearly independent.

Suppose that, for some $i$, $t_i \in \osum\bbP P$. From Lemma \ref{detgroup} we obtain
that $Dt_i = \sum_{v\in P} \gamma_v v$, where $\gamma_v \in \bbZ$.
Moreover, from Lemma \ref{syseq}, $\gamma_v \leq J_0$, if $J_0$ is large enough. Therefore,
$w_2 = w'_2 + Dt_i$, where $w'_2 = \sum_{v\in P} (\beta_v-\gamma_v v)$.
Take $v' = v - Dt_i$; we have $v' = w_0 + w_1 + w'_2 \in W + \osum\bbN \Cout = S$. 
From the induction hypothesis, $v' \in W_1 + \osum\bbN \Cout_0$, and $Dt_i\in\osum\bbN\Cout_0$
from Lemma \ref{detgroup}. Since $v = v'+Dt_i$, we obtain that 
$v \in W_1 + \osum\bbN \Cout_0$.

Otherwise, $\set{t_i:i=1,\ldots,K}$ is disjoint from $\osum\bbP P$. Therefore,
$\osum\bbP{\set{t_i}}$ is disjoint from $\osum\bbP P$ (except 0). By the same arguments as in
the beginning of this proof, there is a hyperplane $\phi \in \Phi_Y$ such that
$\phi(t_i) < 0$ for each $i$, and $\phi(P) \geq 0$.

What is the value of $\phi(v)$? On one hand, we have 
$\phi(v) = \phi(v_0) + D\sum_i \phi(t_i) < \aphibound \reduceto - DK$. On the other hand,
we have $\phi(v) = \phi(w_0 + w_1 + w_2) = \phi(w_0) + \phi(w_1) + \phi(w_2) =
\phi(w_0) + \sum_{v \in \Cout} \alpha_v \phi(v) + \sum_{v \in P} \beta_v \phi(v)
\geq -B \aphibound - J_0 |\Cout| Y \aphibound$.
From applying the triangle inequality to $v=v_0+D\sum_i t_i$, we get that 
$\nsum{v} \leq \nsum{v_0} + DK\alphs\taulimit$, and hence, $K \geq (\mainboundlimit - \reduceto)/D\alphs\taulimit$.
Thus,
$\aphibound H \reduceto - D (\mainboundlimit - \reduceto)/D\alphs\taulimit \geq -B \aphibound
-J_0 |\Cout| Y \aphibound$. This is a
contradiction for a large enough (but still polynomial) $\mainboundlimit$.
\myqed\end{proof}

\begin{proof}[Proof of Lemma \ref{secbound}]
Let $D_i$ be the determinant of $\Cout_i$ (if $\Cout_i$ includes less than $\alphs$
vectors, add unit vectors as usual). Let $D^*$ be the least common multiple of
all $D_i$. We know that $D_i$ is bounded polynomially in $Y$, so $D^*$ is bounded
polynomially by $Y^{|I|}$.

By applying Lemma \ref{reducerlemma} (with $D=D^*$) we get constants $\taulimit$ and $\reduceto$.
We know that each $v \in \Reg(r)$ such that $\nsum{v} \geq \reduceto$ can be written as
$v_0 + D^* t$, where $t \in \tau(r)$. We repeat this construction (adding all the $t$'s
together), until we get $\nsum{v_0} < \reduceto$.

Now, we have to show that $v \in V_i$ iff $v_0 \in V_i$.

Suppose that $v \in V_i$. Thus, we have $v = w_0 + \sum_k \alpha_k P_k$, where
$\alpha_k \geq 0$, and $(P_k)$ are the members of $\Cout_k$. We can write similarly
$v_0 = w_0 + \sum_k \beta_k P_k$. Note that $v$ and $v_0$ are in the same
$(\mainboundlimit,Y)$-region $r$. In particular,  $v-w_0$ and $v_0-w_0$ are on the same side of
each hyperplane going through a member of $L_\mainboundlimit$ and parallel to
$\alphs-1$ members of $\zmint{Y}^\Alphabet$; therefore,
since $\alpha_k \geq 0$, we also have $\beta_k \geq 0$. Moreover, since $v-v_0 = D^* t$,
by Lemma \ref{detgroup}, $v-v_0 \in \osum\bbZ \Cout_i$. Therefore, $\alpha_k-\beta_k$
has to be an integer, and therefore $\beta_k \in \bbN$. Thus, $v_0 \in V_i$.

The proof in the other direction goes in the same way.\myqed
\end{proof}

\begin{proof}[Proof of Lemma \ref{reducerlemma}]
Take $v \in \reg(r)$. From Lemma \ref{regionshape} we know that $v = w + \sum \alpha_i t_i$,
where $\set{t_i: i \in \set{1,\ldots,N}} = T_i \subseteq \zmint\taulimit^\Alphabet$,
and $w \in \zmint\regbaselimit^\Alphabet$. By taking a large enough $\reduceto$, we can ensure that
if $\nsum{x} \geq \reduceto$, then some $\alpha_i$ is greater than $D$. We have that
$v_0 = v-D t_i$ is a proper convex combination of $v$ and $v_1 = v-\alpha_i t_i$, which
are both in $\reg(r)$, thus $v_0$ is also in $\reg(r)$. And if $v \in \Reg(r)$, then
so is $v_0$ ($\phi(v_0-l)$ from the definition of $\reg(r)$ is a proper convex
combination of $\phi(v-l)$ and $\phi(v_1-l)$, $\phi(v-l)$ has the correct sign,
and $\phi(v_1-l)$ either has the correct sign or is 0).
\myqed\end{proof}

\section{A hard grammar}\label{hardgrammar}
As mentioned in the introduction to Section \ref{seccommg}, for $\alphs>2$ 
$\Parikh(G)$ is not a union of polynomial number of simple bundles, which makes it
impossible to use the simple reasoning mentioned. In this
section, we will show an example of a grammar over an alphabet with three symbols
where this is the case. We believe this example is interesting in its own right.

The proof has two parts. In the first part, we create a sequence of grammars $(G_n)$
over four terminal symbols ($x$, $y$, and two temporary ones), such that, if we
ignore the two temporary terminal symbols, the convex hull of $\Parikh(G_n)$ is a
bounded polygon with $2^n$ vertices. It is then straightforward to create a grammar
over three terminal symbols ($x$, $y$, $z$) which has the required property.

\begin{theorem}\label{hardgrammartheorem}
There is a positive grammar $G_n$ of linear size with terminal symbols
$x$, $y$, $S_{n+1}$ and $A_n$, and non-terminal symbols $S_0$ and $X_n$
(where $S_0$ is initial),
such that the vertices of the convex hull of
$\Parikh(G_n)$ are exactly the points $y^i x^{i(i+1)/2} S_{n+1}^{2i+1} A_n^{2N-2i-1}$
for $i = 0, \ldots, N-1$, and $X_n$ generates only $x^N$, where $N=2^n$.
\end{theorem}

\begin{proof}
The proof is by induction.
For $n=0$ we simply put $\mv{S_0}{}{S_1 A_0}$, $\mv{X_0}{}{x}$.

Take the grammar $G_n$; we will construct the grammar $G_{n+1}$.
The symbol $X_{n+1}$ is obtained by the following rule: $\mvm{X_{n+1}}{}{X_n^2}$.

We replace the terminal $A_n$ with a non-terminal with two transition rules:
\begin{itemize}
\item (1) $\mvm{A_n}{}{A_{n+1}^2}$
\item (2) $\mvm{A_n}{}{S_{n+2}^2X_ny}$
\end{itemize}

The new vertices of the convex hull will be obtained by choosing one of the
two rules for $A_n$, and applying it consistently. Thus, 
from each vertex $y^i x^{i(i+1)/2} S_{n+1}^{2i+1} A_n^{2N-2i-1}$ of the
convex hull of $\Parikh(G_n)$, we get the following vertices:

\begin{itemize}
\item (1) $y^i x^{i(i+1)/2} S_{n+1}^{2i+1} A_{n+1}^{4N-4i-2}$
\item (2) $y^{2N-1-i} x^{i(i+1)/2+N(2N-2i-1)} S_{n+1}^{2i+1} S_{n+2}^{4N-4i-2}$
\end{itemize}

Now, we replace the terminal $S_{n+1}$ with a non-terminal with a single
transition rule $S_{n+1} \ra A_{n+1} S_{n+2}$. Hence, we get the following vertices:

\begin{itemize}
\item (1) $y^i x^{i(i+1)/2} S_{n+2}^{2i+1} A_{n+1}^{4N-2i-1}$
\item (2) $y^{2N-1-i} x^{i(i+1)/2+N(2N-2i-1)} A_{n+1}^{2i+1} S_{n+2}^{4N-2i-1}$
\end{itemize}

By taking $j = 2N-1-i$, we can rewrite the second row of vertices as follows:

\begin{itemize}
\item (1) $y^i x^{i(i+1)/2} S_{n+2}^{2i+1} A_{n+1}^{4N-2i-1}$
\item (2) $y^j x^{j(j+1)/2} S_{n+2}^{2j+1} A_{n+1}^{4N-2j-1}$
\end{itemize}

Note that $i$ runs from $0$ to $N-1$, and $j$ runs from $N$ to $2N-1$. We can thus
combine our two rows of vertices into one, indexed by $i$ running from $0$ to $2N-1$,
and thus obtain the induction thesis.\myqed
\end{proof}

Now, let $G_n'$ be obtained from $G_n$ by removing all the terminals $A_{n+1}$ and $S_{n+2}$.
The convex hull of $\Parikh(G_n')$ is the polygon $P$ 
with vertices $x^i y^{i(i+1)/2}$, for $i \in [0..N-1]$.

Now, we create a new grammar $G_n^\circ$ over ${x,y,z}$ whose Parikh image will
be geometrically the infinite cone with base $P$.
This is done simply by adding a new initial symbol $S^\circ$ to $G_n'$,
with two
transitions $\mvm{S^\circ}{z}{S_0S^\circ}$ and $\mvm{S^\circ}{0}{0}$.
It can be easily seen that periods
are single exponential, and since $\Parikh(G^\circ_n)$ is an infinite cone with $2^n$ edges, 
we cannot write $\Parikh(G^\circ_n)$ as a union of a polynomial
number of simple bundles -- indeed, each simple bundle can cover only at most 3 edges.

\section{Complexity results}\label{seccomplexity}

In this section, we provide the tight complexity bounds for the pproblems
we are considering. Some of the results have been previously known (e.g.,
\cite{hyunh1, espfund}), we include all the proofs for completeness.

\begin{theorem}\label{memberregular}
The following membership problem in commutative regular grammars is in P for a fixed
$\Alphabet$:

Given: a regular grammar $G$, $v \in \bbZ^\Alphabet$

Decide whether $v \in \out(G)$
\end{theorem}

To prove this Theorem, we will need the following two lemmas:

\begin{lemma}\label{findruns}

For $P \subseteq \States$, let 
\[\RunsTable_n(P,q) = \set{\out(R): R\mbox{ is a run from }q\mbox{, }|R| \leq n\mbox{, and }P \subseteq \supp(R)}.\]
For a regular grammar $G$ over $\Alphabet$ of size $\alphs$ with $\qstate$ nonterminals, 
the sets $\RunsTable_n(P,q)$ for all $P$ of cardinality at most $k$ and
all $n \leq B$ can be calculated in time $O((2B)^\alphs |G| \qstate^{k+1})$ and space
$O((2B)^\alphs \qstate^{k+1})$.
\end{lemma}

\begin{proof}
Note that a run of length $n$ from $q$ consists of a transition from $q$ to $[q']$ and
a run from $q'$ of length $n-1$. Thus, we get the following recursive formula:

\begin{eqnarray*}
\RunsTable_1(P,q) &=& \set{\out(\trans): \trans \in \Trans, \source(\trans) = q, \target(\trans) = 0, \{q\} \subseteq P } \\
\RunsTable_n(P,q) &=& \RunsTable_{n-1}(P,q) \cup \RunsTable_n(P-\set{q},q) \cup \\&& \set{\out(\trans) + \RunsTable_{n-1}(P,r): \trans \in \Trans, \target(\trans) = [r], \source(\trans) = q}
\end{eqnarray*}
We know that $\RunsTable_n(P,q) \subseteq \zmint{n}^\Alphabet$. This allows us to calculate all the sets
$\RunsTable_n(p,q)$ using dynamic programming.\myqed
\end{proof}

\begin{lemma}\label{findpaths}
Let 
\[\PathTable_n(q_1,q_2) = \set{\out(R): R\mbox{ is a path from $q_1$ to $q_2$ and }|R| \leq n}.\]
For a regular grammar $G$ over $\Alphabet$ of size $\alphs$ with $\qstate$ nonterminals, the sets $\PathTable_n(q_1, q_2)$ for
all $n \leq B$ can be calculated in time $O((2B)^\alphs |G| \qstate^2)$
and space $O((2B)^\alphs \qstate^{k+2})$.
\end{lemma}

\begin{proof}
The idea is essentially the same as in Lemma \ref{findruns}. We use the following recursive formula:

\begin{eqnarray*}
\PathTable_0(q_1,q_2) &=& \emptyset \mbox{ if }q_1 \neq q_2 \\
\PathTable_0(q, q) &=& \set{0} \\
\PathTable_n(q_1,q_2) &=& \PathTable_{n-1}(P,q) \cup \\&&
\set{\out(\trans) + \PathTable_{n-1}(r,q_2): \trans \in \Trans, \target(\trans) = [r], \source(\trans) = q_1}
\end{eqnarray*}
\myqed
\end{proof}

\begin{proof}[Proof of Theorem \ref{memberregular}]

The algorithm is as follows.
\begin{enumerate}
\item Let $\frunbound$ be the bound on $R_1$ from Theorem \ref{decomposition}.

\item Using Lemma \ref{findruns} we find the sets $\RunsTable_{\frunbound}(P,\qini)$ for all $P$ such that
$|P| \leq \alphs$.

\item Using Lemma \ref{findpaths} we find the sets $\PathTable_\qstate(q,q)$ for all $P$ such that
$|P| \leq \alphs$.

\item For each $P$ such that $|P| \leq \alphs$:

\ind For each linearly independent subset $Z$ of $\cup_{q \in P} \PathTable_\qstate(q,q)$:

\ind \ind For each element $w$ of $\RunsTable_{\frunbound}(P,\qini)$:

\ind \ind \ind if $v-w \in \osum\bbN Z$:

\ind \ind \ind \ind return YES.

\item Otherwise return NO.\myqed
\end{enumerate}\medskip

\noindent Since $Z$ is linearly independent, we can check whether $v-w \in \osum\bbN Z$ using Gaussian elimination.

It is straightforward from Theorem \ref{decomposition} that this algorithm will return
YES if $v \in \out(G)$. It is also straightforward that $v \in \out(G)$ if the algorithm
answers YES.

This algorithm runs in time $O\left((2\frunbound)^\alphs |G| \qstate^{\alphs+1} + \qstate^{\alphs+1}\right)$.
It uses space $O((2B)^\alphs \qstate^{k+2})$.
\myqed
\end{proof}

\begin{theorem}\label{membercf}
The membership problem in commutative grammars is in NP.
\end{theorem}
This result was previously known \cite{espfund}.

\begin{proof}
From Theorem \ref{decomposition} we know that if $v \in \Parikh(G)$, then
$v = \Parikh(R_1) + \sum \Parikh(C_k) n_k$, where $R_1$ and $C_k$ are bounded
exponentially, and $C_k$ are linearly independent. We can simply guess $R_1$ and $C_k$.\myqed
\end{proof}

\begin{theorem}\label{inclcf}
The inclusion, universality, and disjointness problems in commutative grammars over
an alphabet of fixed size are in \PiP.
\end{theorem}
\begin{proof}
We will consider inclusion; the other problems can be solved in the same way.

Let $G_1$ and $G_2$ be two commutative grammars in normal form 
with at most $N$ non-terminals each, over the same fixed alphabet $\Alphabet$.

By Theorem \ref{normalcommcorol}, there exists a number $\grammarbound$ which is single
exponential in $N$, such that $\Parikh(G_1) \subseteq \Parikh(G_2)$ iff
$\Parikh(G_1) \cap [-\grammarbound..\grammarbound]^\Alphabet \subseteq \Parikh(G_2) \cap [-\grammarbound..\grammarbound]^\Alphabet$.
Thus, we need to check the membership in $\Parikh(G_2)$ for all the elements of
the set $[-\grammarbound..\grammarbound]^\Alphabet \cap \Parikh(G_1)$. This can be done in \PiP.
\myqed
\end{proof}

\begin{theorem}\label{inclreg}
The inclusion and universality problems in regular grammars over an alphabet of fixed size are in coNP.
\end{theorem}

\begin{proof}
The proof is just like for Theorem \ref{inclcf}. Again, we need to check the membership
in $\Parikh(G_2)$ for all the elements $v$ of
$[-\grammarbound..\grammarbound]^\Alphabet \cap \Parikh(G_1)$. However, this time
checking membership of $v$ can be done in P, so the whole algorithm works in coNP.
\end{proof}

\begin{theorem}\label{incpicomplete}
The problem of inclusion of commutative grammars over $\Alphabet=\set{a}$ is \PiP hard.
\end{theorem}

This result was previously known \cite{hyunh1}.

\begin{proof}
We will reduce the following problem ($\mbox{\bf 3-CNF-QSAT}_2$). Let the clauses $C_j$
$(0\leq j<m)$ be
disjunctions of at most three literals of form $x_i$, $\neg x_i$, $\neg y_i=$ or $\neg y_i$.
Does the formula 
\[
\forall x_0 \forall x_1 \ldots \forall x_{k-1} 
\exists y_0 \exists y_1 \ldots \exists y_{l-1}
\bigwedge_{j<m} C_j
\] hold? Quantifiers run over two possible values of each variable (either $x_j$ or $\neg x_j$ is true).

The symbol $A_i$ $(i<k)$ generates $a^{2^i}$, for $i<k$. This can be realized with the
following rules: $A_0 \ra a$, $A_i \ra A_{i-1}A_{i-1}$ for $i>0$.

The symbol $A_i^?$ $(i<k)$ generates $a^{2^i}$ or nothing. This can be realized with the following
rules: $A_i^? \ra 0$, $A_i^? \ra A_i$.

The symbol $C_j$ $(j<m)$ generates $(a^{2^k})^{4^j}$. This can be realized with the following
rules: $C_j \ra A_{k-1}A_{k-1}$, $C_j \ra C_{j-1}C_{j-1}C_{j-1}C_{j-1}$ for $j>0$.

The symbol $C_j^?$ $(j<m)$ generates $(a^{2^k})^{4^j}$ or nothing. This can be realized with
the following rules: $C_j^? \ra 0$, $C_j^? \ra C_j$.

The symbol $X_i$ $(i<k)$ has the following two rules: 
$X_i \ra A_i \prod_{j<m}(C_j^?: x_i \in C_j) |\prod_{j<m}(C_j^?: \neg x_i \in C_j)$.

The symbol $Y_i$ $(i<l)$ has the following rules: 
$Y_i \ra \prod_{j<m}(C_j^?: y_i \in C_j) |\prod_{j<m}(C_j^?: \neg y_i \in C_j)$.

The symbol $S_1$ has the following rules: $S_1 \ra \prod_{i < k} A_i^? \prod_{j<m} C_j$.

The symbol $S_2$ has the following rules: $S_2 \ra \prod_{i < k} X_i \prod_{i<l} Y_i$.

We ask whether the language generated by $S_1$ is a subset of the language generated
by $S_2$.

Note that each $A_i$ is generated at most once, and each $C_j$ is generated at most three
times. The definitions of these symbols ensure that we can treat them as independent:
for each combinations of $A_i$ and $C_j$ in $S_1$, we need to find the same combination
in $S_2$.

The language generated by $S_1$ has $2^k$ elements (for each $A_i$, we can either add
it or not). These correspond to possible valuations of variables of $x_i$. To match
the given element of $S_1$ in $S_2$, we need to make the same choices in $X_i$. We also
need to generate $C_0 C_1 C_2 \ldots C_{j-1}$. This means that we have to make choices
in $Y_i$ which generate all the clauses which were not covered by our choices in $X_i$.
Therefore, $S_1 \subseteq S_2$ iff the formula is true.
\end{proof}

\begin{theorem}
The problem of universality (is $\Psi(G) = \bbZ^\Alphabet$?) of commutative grammars over $\Alphabet=\set{a}$ is \PiP hard.
\end{theorem}

\begin{proof}
We use the same reduction as in Theorem \ref{incpicomplete}. It is enough to have a symbol
$S_3$ which generates the complement of $S_1$. Indeed, add a new symbol $S_4$, with
rules $S_4 \ra S_2 | S_3$. The universality of $S_4$ is equivalent to $S_1$ (the
complement of $S_3$) being included in $S_2$, which is equivalent to our
$\mbox{\bf 3-CNF-QSAT}_2$ formula being true.

This can be done as follows:

The symbol $C_j^H$ generates 0, 2, or 3 copies of $C_j$. This can be done with
the following rules: $C_j^H \ra 0 | C_jC_j | C_jC_jC_j$.

The symbol $Z^+$ generates any number of $a$'s which is at least $2^k4^m$. This
can be realized as follows: $Z^+ \ra C_{m-1}C_{m-1}C_{m-1}C_{m-1} | aZ^+$.

The symbol $Z^-$ generates any negative number of $a$'s. This can be realized as follows:
$Z^- \ra a^{-1}Z^- | a^{-1}$.

Now, we can define $S_3$ as $S_3 \ra Z^+ | Z^- | \prod_{i < k} A_i^? \prod_{j<m} C_j^H$.
\end{proof}

\begin{proposition}
The problem of membership in commutative grammars over $\Alphabet=\set{a}$ is NP hard.
\end{proposition}

\begin{proof}
We reduce the $\mbox{\bf 3-CNF-SAT}$ problem. The proof is the same as in
Theorem \ref{incpicomplete}, except that we take $k=0$.
\end{proof}

\begin{proposition}\label{uninonreg}
Let $G$ be a commutative regular grammar over $\Alphabet = \{a\}$.
Then the problem of deciding universality ($\prod(G) = \bbN^\Alphabet$)
is coNP-hard.
\end{proposition}

\begin{proof}
We reduce the 3CNF-SAT problem. Let $\phi = \bigwedge_{1\leq i\leq k} C_i$ be a
3CNF-formula with $n$ variables $x_1\ldots x_n$ (which can be 0 or 1)
and $k$ clauses. Let $p_1, p_2, \ldots, p_n$
be $n$ distinct prime numbers. Let $i \in [1..k]$. Suppose that clause $C_i$
is of form $\bigvee_{k\in[1..3]} x_{a_k}=v_{a_k}$. Our grammar will have states
$S^i_j$, where $0 \leq j < M_i = p_{a_1} p_{a_2} p_{a_3}$; we have cyclic
transitions $\mvm{S^i_j}{a}{S^i_{(j+1) mod M_i}}$, and $\mvm{S^i_j}{0}{0}$ for each 
$j$ not satisfying $\bigvee_{k\in[1..3]} j \bmod p_{a_k}=v_{a_k}$. We also have 
transitions $\mvm{s_0}{0}{S^i_0}$ for each $i$.

From simple number theoretic arguments we get
that $x \notin \prod(G)$ iff the formula $\phi$ is satisfied for $x_i = x
\bmod p_i$.
\myqed\end{proof}

\begin{proposition}\label{membunfix}
For a commutative regular grammar 
$G$ over $\Alphabet$ (whose size is not fixed), and $K \in \bbN^\Alphabet$,
the problem of deciding whether $K \in \prod(G)$ is NP-hard.
\end{proposition}

\begin{proof}
We show a reduction from the Hamiltonian
circuit problem. Let $(V,E)$ be a graph. We take $\Alphabet = Q = V$, and 
for each edge $(v_1, v_2)$ we add a transition $\mvm{v_1}{v_2}{v_2}$.
We pick an initial state $s_0$ and add a final transition $\mvm{s_0}{0}{0}$.
The graph $(V,E)$ has a Hamiltonian circuit iff $K = (1,1,\ldots) \in \Parikh(G)$.
\myqed\end{proof}

\begin{proposition}\label{disunfix}
The disjointness problem in commutative regular grammars (over
an alphabet of unfixed size) is coNP-complete.
\end{proposition}
\begin{proof}
From Proposition \ref{membunfix} we know that membership is NP-hard. We can easily
construct a grammar $G_2$ such that $\Parikh(G_2) = \{K\}$ (where $K$ is from
the proof of Proposition \ref{membunfix}), and ask for non-disjointness
of $G$ and $G_2$.

From Theorem \ref{membercf} we also know that membership is in NP (even for
non-regular grammars). We reduce non-disjointness of $G_1$ and $G_2$ to membership
in the following way. It is sufficient to check whether $0 \in \Parikh(G_1) - \Parikh(G_2)$,
where $\Parikh(G_1) - \Parikh(G_2) = \{v_1-v_2: v_1 \in \Parikh(G_1), v_2 \in \Parikh(G_2)\}$.
We create a grammar $G_1-G_2$ such that $\Parikh(G_1-G_2) = \Parikh(G_1) - \Parikh(G_2)$:
first, create the grammar $-G_2$ such that $\Parikh(-G_2) = \{-v: v \in \Parikh(G_2)$
by replacing each production by its negative, and then create the grammar $G_1-G_2$
by replacing each final transition in $G_1$ (i.e., $\delta$ such that $\target(\delta)=0$)
by transitions going to each initial state of $G_2$.\myqed
\end{proof}

The following table summarizes the complexities of various problems regarding
commutative grammars. Alphabet size F means fixed, and U means unfixed.
We include the very recent result regarding unfixed alphabets \cite{haase}
that inclusion for regular grammars over an unfixed alphabet is coNEXP hard; 
it is likely that the proof can be also adapted for universality.
On the other hand, it
is known from \cite{semipi2} that inclusion and universality for context-free 
grammars over an unfixed alphabet is in coNEXP.
For completeness, we also include N (the non-commutative case -- note that the
membership problem is actually a different problem in the non-commutative case, since we cannot 
encode long words succinctly with the length of the word in binary); these results are known
from other sources (\cite{universalityregular}, also see \cite{hopcroft79} for a reference). The letter c means
complete. Note that inclusion and equivalence problems easily reduce to each other.

\begin{center}
\begin{tabular}{|c|ccccc|}
 \hline
\multicolumn{6}{|c|}{regular languages} \\
 \hline
alphabet size & 1 & 2 & F & U & N               \\
 \hline
membership    & P & P & P & NPc & P      \\
universality  & coNPc & coNPc & coNPc & ? & PSPACEc  \\
inclusion     & coNPc & coNPc & coNPc & coNEXPc & PSPACEc   \\
disjointness  & P  & P  & P & coNPc & P  \\
\hline
\multicolumn{6}{|c|}{context-free languages} \\
 \hline                                              
alphabet size & 1 & 2 & F & U & N \\                     
 \hline                                              
membership    & NPc & NPc & NPc & NPc & P \\
universality  & \PiP c & \PiP c & \PiP c & ? & undecidable \\
inclusion     & \PiP c & \PiP c & \PiP c & coNEXPc & undecidable \\
disjointness  & coNPc & coNPc & coNPc & ? & undecidable \\
\hline
\end{tabular}
\end{center}

\section{Conclusion}\label{secconclusion}

The table above contains question marks for languages of unfixed size. It is known
that these problems are \PiP-hard for context-free grammars, and they are in NEXP. Our method
heavily uses the fact that the size of the alphabet is fixed, so we probably
cannot easily generalize it. As mentioned above, coNEXP hardness of inclusion
for regular grammars has been shown recently \cite{haase}.

A natural question extending this research is counting. We can answer the question
whether $v \in \Parikh(G)$, but what about the number of paths (or words) which lead
to the given $v$? This number is exponential in $\nsum{v}$, and thus it could be
very large, so we cannot always hope for the exact answer---but we could count up
to some threshold $T$, modulo $M$, or count approximately.
We can answer the question whether $\Parikh(G_1) \subseteq \Parikh(G_2)$,
but is the number of paths smaller in the first case than in the second case?

Thanks to everyone on AUTOBÓZ 2009 for the great atmosphere of research,
especially to Sławek Lasota for introducing me to these problems. Also I would
like to thank Anthony Widjaja Lin for our collaboration on the merged paper
\cite{licsfull}, and the anonymous referees for their insightful comments.
This work is supported by the Polish National Science Centre Grant
DEC - 2012/07/D/ST6/02435.

\bibliographystyle{alpha}
\bibliography{parikh}{}

\end{document}